\documentclass[journal]{IEEEtran}

\usepackage[bookmarks=false]{hyperref}
\hypersetup{
    colorlinks=true,
    linkcolor=black,
    filecolor=black,
    urlcolor=black,
}

\usepackage{tikz}
\usetikzlibrary{patterns}
\usetikzlibrary{positioning,arrows.meta}
\usetikzlibrary{shapes.geometric, arrows.meta, positioning, calc}
\usepackage{algpseudocode}
\usepackage[utf8]{inputenc}
\usepackage[T1]{fontenc}
\usepackage{algorithm} 
\usepackage{nomencl}
\usepackage[normalem]{ulem}
\usepackage{etoolbox}
\usepackage[dvipsnames]{xcolor}
\usepackage{wrapfig}
\usepackage{graphicx}
\usepackage{makecell}
\graphicspath{ {./images/} }

\newcommand{\maybe}[1]{}
\definecolor{salmon}{RGB}{255,111,105}

\usepackage{tikz}
\usetikzlibrary{matrix}

\makeatletter
\newdimen\multi@col@width
\newdimen\multi@col@margin
\newcount\multi@col@count
\tikzset{
  multicol/.code={%
    \global\multi@col@count=#1\relax
    \global\let\orig@pgfmatrixendcode=\pgfmatrixendcode
    \global\let\orig@pgfmatrixemptycode=\pgfmatrixemptycode
    \def\pgfmatrixendcode##1{\orig@pgfmatrixendcode%
      ##1%
      \pgfutil@tempdima=\pgf@picmaxx
      \global\multi@col@margin=\pgf@picminx
      \advance\pgfutil@tempdima by -\pgf@picminx
      \divide\pgfutil@tempdima by #1\relax
      \global\multi@col@width=\pgfutil@tempdima
      \pgf@picmaxx=.5\multi@col@width
      \pgf@picminx=-.5\multi@col@width
      \global\pgf@picmaxx=\pgf@picmaxx
      \global\pgf@picminx=\pgf@picminx
      \gdef\multi@adjust@position{%
        \setbox\pgf@matrix@cell=\hbox\bgroup
        \hfil\hskip-\multi@col@margin
        \hfil\hskip-.5\multi@col@width
        \box\pgf@matrix@cell
        \egroup
      }%
      \gdef\multi@temp{\aftergroup\multi@adjust@position}%
      \aftergroup\multi@temp
    }
    \gdef\pgfmatrixemptycode{%
      \orig@pgfmatrixemptycode
      \global\advance\multi@col@count by -1\relax
      \global\pgf@picmaxx=.5\multi@col@width
      \global\pgf@picminx=-.5\multi@col@width
      \ifnum\multi@col@count=1\relax
       \global\let\pgfmatrixemptycode=\orig@pgfmatrixemptycode
      \fi
    }
  }
}
\makeatother
\usepackage{blindtext}
\usepackage{threeparttable} 
\usepackage{multicol}
\ifCLASSINFOpdf
  \graphicspath{{../pdf/}{../jpeg/}}
  \DeclareGraphicsExtensions{.pdf,.jpeg,.png}
\else

\fi

\usepackage{amsmath}
\usetikzlibrary{
positioning,
calc,
fit,
arrows.meta,
backgrounds,
decorations.pathreplacing
}

\usepackage{cite}
\usepackage{amsmath,amssymb,amsfonts}
\usepackage{algorithm}
\usepackage{acro}
\usepackage{algpseudocode}
\usepackage{graphicx}
\usepackage{textcomp}
\usepackage{xcolor}
\usepackage{tikz}
\usepackage{adjustbox}
\usepackage[version=4]{mhchem}
\usepackage{booktabs}
\usepackage{multirow}
\usepackage{siunitx}
\usepackage{eurosym}
\usepackage{mhchem}
\usepackage{amssymb}
\usepackage{tikz}
\usepackage{pgfplots}
\usepackage{multirow}
\usepackage{graphicx}
\usepackage{hyperref}
\usepackage{eurosym}
\usepackage{subcaption}
\def\BibTeX{{\rm B\kern-.05em{\sc i\kern-.025em b}\kern-.08em
    T\kern-.1667em\lower.7ex\hbox{E}\kern-.125emX}}
\usepackage{url}

\usepackage[inkscapelatex=false]{svg}

\usepackage{amsthm}
\newtheorem{theorem}{Theorem}[section] 
\newtheorem{definition}{Definition} 
\newtheorem{remark}{Remark}
\newtheorem{lemma}[theorem]{Lemma}
\newtheorem{corollary}[theorem]{Corollary}
\newtheorem{proposition}[theorem]{Proposition}

\DeclareAcronym{is}{
  short=IS,
  long=In-sample,
}
\DeclareAcronym{amc}{
    short=AMC,
    long= Additivity of marginal contributions at the top level,
}
\DeclareAcronym{brp}{
    short=BRP,
    long= Balance Responsible Party,
}
\DeclareAcronym{ppa}{
    short=PPA,
    long= Power Purchase Agreement,
}
\DeclareAcronym{wmape}{
    short = WMAPE,
    long = Weighted Mean Absolute Percentage Error
}
\DeclareAcronym{vcg}{
    short = VCG,
    long = Vickrey-Clarke-Groves
}

\begin{document}
%

\title{Sharing the Gains of Aggregation: \\ Cooperative Imbalance Cost Allocation}


\author{Asmus W. Eriksen and Jalal Kazempour, \textit{Senior~Member,~IEEE}

\vspace{-5mm}

\thanks{
Asmus W. Eriksen and Jalal Kazempour are with the Department of Power and Energy Systems, Technical University of Denmark, 2800 Kgs. Lyngby, Denmark (e-mail: \{aswin, jalal\}@dtu.dk)
}
}

    




\vspace{-5mm}
\maketitle


\IEEEaftertitletext{\vspace{-0.8\baselineskip}}
\maketitle
\thispagestyle{plain}
\pagestyle{plain}
\begin{abstract}
A \textit{facilitator} is an intermediary that offers renewable producers and consumers fixed-price contracts and, acting as their balance responsible party, manages the residual  imbalances in the market. Pooling imperfectly correlated residuals nets consumers' imbalances  and reduces the portfolio's total imbalance cost, but raises an allocation question: how should  these savings be divided among heterogeneous consumers? We formulate this as a  cooperative game, the \textit{imbalance netting game}, and compare six allocation mechanisms under a two-price imbalance settlement in terms of computational requirements, budget balance, group rationality, and additivity. We establish three analytical results: a necessary and sufficient condition for budget balance of the marginal cost contribution mechanism, group rationality of the marginal cost contribution and Vickrey-Clarke-Groves mechanisms, and a necessary and sufficient condition for when the Gately point is well-defined. On Danish 2024 data for 19 consumers, aggregation reduces imbalance cost by 10\%, but how these savings are shared depends strongly on the allocation rule. A flat-rate allocation proportional to consumption, used as a benchmark for socialized imbalance pricing, charges some consumers more than twice their standalone cost, whereas all the game-theoretic mechanisms studied produce stable allocations in practice. 
\end{abstract}



\section{Introduction}\label{sec:introduction}
\IEEEPARstart{A}{s} electricity markets liberalize and the share of variable renewable generation increases, new intermediaries are emerging to shield consumers and renewable producers from market and balancing risk. We refer to such an intermediary as a \textit{facilitator}. The facilitator negotiates a \ac{ppa} between a renewable power producer and one or more consumers. In addition, the facilitator manages the \textit{residual}, i.e., the part of each consumer's net position not covered by its \ac{ppa}, which is uncertain because both production and consumption are uncertain. A concrete example, and the motivation for this work, is the Danish supplier Reel\footnote{\url{https://reel.energy}}, which gives a portfolio of producers and consumers a fixed price and manages the residual balancing risk of the portfolio.

The facilitator therefore has two roles: brokering the \ac{ppa}s and managing the residual. In this paper, the \ac{ppa} contracts are taken as fixed and given, and we focus on the allocation of the residual imbalance cost. As illustrated in Fig.~\ref{fig:facilitator}, each consumer holds a fixed-price \ac{ppa} covering a share of renewable production, while the remaining net position must be traded in the market. Because both demand and production are uncertain, this position carries forecast risk. Acting as the \ac{brp} for the portfolio, the facilitator bids the aggregated residual in the day-ahead market, settles the resulting deviations in the balancing market, and prices the \textit{imbalance cost} back to its consumers. By doing so, the facilitator improves the bankability of PPAs for smaller consumers and projects, and offers consumers a hedge against balancing risk they may not manage on their own.

The consumer's energy bill, excluding taxes and grid tariffs, has four components: A fixed \ac{ppa} price, a fixed service fee, day-ahead market charges for any residual consumption or generation, and an imbalance charge, i.e., the consumer's share of the portfolio's imbalance cost. The first two are set in advance, and the third is settled directly at the day-ahead price. Only the imbalance charge and its division among consumers are studied here, with the \ac{ppa} price and the service fee taken as given.

\begin{figure}[t]
    \centering
    \begin{tikzpicture}[
  scale=0.55,
  >=Stealth,
  box/.style={draw, line width=0.6pt, fill=white, align=left, font=\scriptsize, inner sep=5pt},
  every node/.style={font=\scriptsize}
]

\node[box, anchor=west, minimum width=2.2cm] (ren) at (0,6.35) {Renewables};
\node[box, anchor=west, minimum width=2.2cm] (dam) at (0,3.8) {Day-ahead\\ [-1pt] \hspace{1mm} market};
\node[box, anchor=west, minimum width=2.2cm] (bm)  at (0,1.9) {Balancing\\[-1pt] \hspace{0.8mm} market};

\node[box, anchor=west, minimum width=2.2cm, minimum height=0.7cm] (fac) at (7.85,1.9) {Facilitator};

\coordinate (facSegL) at ($(fac.south)+(0.55,0.06)$);
\draw[salmon, line width=0.65pt] (facSegL) -- ($(fac.south east)+(-0.03,0.06)$);
\draw[salmon, line width=0.65pt]
  ($(facSegL)+(0.4,0.0)$)
    to[out=55,in=180] ($(facSegL)+(0.8,0.65)$)
    to[out=0,in=125]  ($(facSegL)+(1.2,0)$);
    
\draw[salmon, line width=0.6pt, ->] (fac.south east) ++(-0.6,-0.05) -- ++(0,-0.5);
\node[salmon, anchor=north west, font=\scriptsize] at ($(fac.south east)+(-3.0,-0.4)$)
  {Aggregated distribution:};
\node[salmon, anchor=north west, font=\scriptsize] at ($(fac.south east)+(-4.5,-1)$)
  {Netting decreases the total imbalance};

\node[box, anchor=west, minimum width=2.0cm, minimum height=0.6cm] (cons1) at (8.05,5.55) {};
\node[box, anchor=west, minimum width=2.0cm, minimum height=0.6cm] (cons2) at (7.8,5.95) {};
\node[box, anchor=west, minimum width=2.0cm, minimum height=0.6cm] (cons3) at (7.55,6.35) {};

\node[anchor=north west, font=\scriptsize] at ($(cons3.north west)+(0.15,-0.0)$) {Consumer 1};

\coordinate (Lcons1) at ($(cons1.south)+(0.35,0.06)$);
\draw[salmon, line width=0.6pt] (Lcons1) -- ($(cons1.south east)+(-0.03,0.06)$);
\draw[salmon, line width=0.6pt]
  ($(Lcons1)+(0.35,0)$)
    to[out=12,in=180] ($(Lcons1)+(1.15,0.3)$)
    to[out=0,in=95]  ($(Lcons1)+(1.30,0)$);

\coordinate (Lcons2) at ($(cons2.south)+(0.35,0.06)$);
\draw[salmon, line width=0.6pt] (Lcons2) -- ($(cons2.south east)+(-0.03,0.06)$);
\draw[salmon, line width=0.6pt]
  ($(Lcons2)+(0.419,0)$)
    to[out=60,in=180] ($(Lcons2)+(0.719,0.3)$)
    to[out=0,in=120]  ($(Lcons2)+(1.019,0)$);

\coordinate (Lcons3) at ($(cons3.south)+(0.35,0.06)$);
\draw[salmon, line width=0.6pt] (Lcons3) -- ($(cons3.south east)+(-0.03,0.06)$);
\draw[salmon, line width=0.6pt]
  ($(Lcons3)+(0.20,0)$)
    to[out=95,in=180] ($(Lcons3)+(0.40,0.3)$)
    to[out=0,in=168]  ($(Lcons3)+(1.20,0)$);

\draw[salmon, line width=0.6pt, ->] (cons1.south east) ++(-0.05,0.15) -- ++(0.75,1.7);
\draw[salmon, line width=0.6pt, ->] (cons2.south east) ++(-0.05,0.15) -- ++(0.85,1.4);
\draw[salmon, line width=0.6pt, ->] (cons3.south east) ++(-0.05,0.15) -- ++(0.95,1.1);

\node[salmon, anchor=south west, font=\scriptsize] at ($(cons3.north east)+(-2.8,0.15)$)
  {Individual imbalance distributions};

\draw[line width=1pt, ->] (ren.east) -- (cons3.west);
\node[anchor=north, align=center] at ($(ren.east)!0.5!(cons1.west)+(-0.4,1.2)$) {Fixed PPA};

\draw[line width=1pt, <->] (fac.north) -- (cons1.south);
\node[anchor=west, align=left] at ($(fac.north)!0.5!(cons1.south)+(0.7,0.3)$) {Residuals};
\node[anchor=west, align=left] at ($(fac.north)!0.5!(cons1.south)+(0.15,-0.3)$) {(deficit/excess)};

\draw[line width=1pt, <->] (fac.north west) -- (dam.east);
\node[anchor=south, align=center] at ($(fac.north west)!0.5!(dam.east)+(0.4,0.08)$) {Residuals};

\draw[line width=1pt, {Stealth}-{Stealth}] (fac.west) -- (bm.east);
\node[anchor=north, align=center] at ($(fac.west)!0.5!(bm.east)+(0,-0.0)$) {Imbalances};

\end{tikzpicture} 
\caption{\small{System overview. Each consumer holds a fixed-price Power Purchase Agreement (PPA), which is a fixed share of renewable production. The facilitator acts as the Balance Responsible Party (BRP) for the resulting residual net positions, bidding to cover expected shortfalls and offering expected surplus in the day-ahead market, then settling deviations in the balancing market. Pooling imperfectly correlated residuals reduces the portfolio's imbalance. Black arrows are power transactions.}}
    \label{fig:facilitator}
    \vspace{-3mm}
\end{figure}
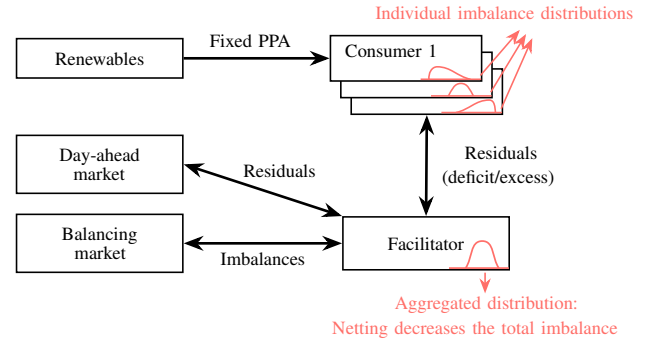

The economic value of this service arises from a demand-side \textit{portfolio effect}. The consumers' residual positions are stochastic and imperfectly correlated, so their forecast errors partly offset when pooled. An imbalance from one consumer may be netted by an opposite imbalance from another, reducing the portfolio's relative imbalance volume and, under suitable pricing, its total imbalance cost. Aggregation can thus cause a  cost saving relative to the sum of the consumers' standalone imbalance costs.

The central question is how this saving should be allocated. A fully socialized scheme, in which all consumers pay the same average imbalance cost per unit of demand, is simple but ignores heterogeneity in forecastability, size, and contribution to portfolio imbalance. Charging every consumer its standalone imbalance cost, conversely, ignores the netting value that exists only through aggregation. Neither extreme is satisfactory for a heterogeneous portfolio, which motivates the research question of this paper: \textit{how should the imbalance cost in an aggregated  portfolio be allocated among consumers to be fair, stable, budget-balanced, and practically implementable?}

We study this question for inelastic consumers in the day-ahead and balancing markets, so that a consumer's imbalance stems from the deviation between its contracted and realized residual volume rather than from any dispatch decision. Imbalances are settled under an asymmetric \textit{two-price} imbalance settlement: deviations that help balance the power system are settled at the day-ahead price, while those that aggravate the system imbalance are settled at a less favorable imbalance price. We adopt this setting because it quantifies directly the value of reducing imbalance volume as an expected cost saving\footnote{Under a \textit{one-price} scheme, aggregation mainly reduces income volatility rather than imbalance cost. The two-price setting is also relevant given regulatory developments: while regulatory requirements are pushing Europe towards symmetric imbalance schemes~\cite{ImbalanceDirective}, the Danish transmission system operator's upcoming move to \textit{full cost balancing} points toward a more asymmetric ``polluter pays'' design~\cite{EnerginetFullCost}, which reintroduces incentives to reduce imbalance volume.}.

The value of aggregation under uncertainty is well established. As renewable penetration grows, market participants face higher uncertainty~\cite{Chu2016}, and aggregating geographically or technologically diverse resources can reduce aggregate variability~\cite{Fleten2018,Tekani2025}. As virtual power plants and other multi-owner aggregations become prevalent~\cite{Oren2025}, allocating the gains of aggregation becomes a central market design problem, for which cooperative game theory offers a natural framework. It has been applied to energy communities~\cite{MITRIDATI2021102177}, shared storage access~\cite{vespermann2021}, and alternatives to marginal pricing in wholesale markets~\cite{Exizidis2019}. Baeyens et al.~\cite{Baeyens2013} use it to share the gains of aggregating wind producers in forward markets.

Uncertainty is increasingly a demand-side phenomenon as well: the growth of distributed generation and intermittent self-sufficiency makes load-related risk a major driver of retailer and \ac{brp} exposure~\cite{Russo2020}. This motivates studying aggregation of consumers' residual demand. Closest to our setting, Guo et al.~\cite{Guo2021} model co-located data centers as a cooperative game with coalition-level re-optimized bidding and an expected-cost marginal allocation whose group rationality holds only asymptotically over repeated realizations. Valencia Zuluaga and Oren~\cite{Zuluaga2025} study a related uniform-price allocation and show that it preserves fairness in peer-to-peer markets. However, neither study systematically compares alternative cooperative allocation mechanisms or characterizes their budget balance and stability under two-price imbalance settlement. This paper addresses these gaps.

This paper makes four distinct contributions. 
First, we cast imbalance cost allocation in a facilitator's portfolio as a cooperative game, the \textit{imbalance netting game}, where coalition costs are determined by the two-price imbalance settlement. This framework makes fairness and stability operationally precise through the core and group rationality conditions.
Second, we systematically compare six allocation mechanisms (Shapley value, marginal cost contribution, \ac{vcg}, nucleolus, marginal price allocation, and Gately point) across four properties: computational tractability, budget balance, group rationality, and additivity. We identify which mechanisms remain practical as portfolio size grows.
Third, we derive analytical results establishing the game-specific properties reported in Table \ref{tab:mechanisms} of Section~\ref{sec:alloc}. Three are of independent interest: ($i$) a necessary and sufficient condition for the marginal cost contribution mechanism to be budget-balanced, showing that when violated, it systematically under-recovers costs and leaves the facilitator in budget deficit. ($ii$) group rationality for both the marginal cost and VCG mechanisms. ($iii$), a condition for when the Gately point is well-defined, with the unique core allocation that applies when it is not.
Fourth, on a real Danish case study of 19 consumers sharing a photovoltaic \ac{ppa} with 2024 data, we quantify how the mechanisms redistribute aggregation gains. 

The remainder of this paper is organized as follows. Section~\ref{sec:model} formulates the imbalance netting game and the pipeline for computing imbalance costs. Section~\ref{sec:alloc} presents the six allocation mechanisms. Section~\ref{sec:analytical} 
provides the analytical results. Section~\ref{sec:results} applies the framework to a real-world case study. Section~\ref{sec:conclusion} concludes.

\vspace{1mm}
\section{The Imbalance Netting Game}\label{sec:model}
We model the cost allocation as a cooperative game \cite{ACourseInCooperativeGameTheory}, in which a set of $|N|$ consumers of the facilitator can form binding coalitions, that is, subsets $S\!\subseteq\!N$ that cooperate to reduce the total procurement cost borne by the group; the full set $N\!=\!\{1,2,\dots,|N|\}$ is the \textit{grand coalition}. The procurement cost has two parts: the cost of energy bought in the day-ahead market, and the imbalance cost settled in the balancing stage. We write $N\!\setminus\!\{i\}$ for the coalition of all consumers except $i$.
\begin{remark}
Consumers pay the day-ahead price for electricity plus an imbalance charge. The day-ahead procurement cost is always covered by the consumer and only the day-ahead/balancing price spread is exposed to the facilitator. The game therefore reduces to allocating the imbalance cost alone.
\end{remark}

Each coalition is assigned a value by a characteristic function $c: 2^{{N}}\!\to\!\mathbb{R}$, and here $c(S)$ is the imbalance cost incurred by coalition $S$, so that $c(S)\!\ge\!0$ for every coalition $S$, as shown in Section \ref{sec:cost}. Of particular interest is the \textit{grand-coalition  cost} $c(N)$, the total imbalance cost of the fully aggregated portfolio. 
We now cast the allocation of imbalance costs as a cooperative game.

\begin{definition}[Imbalance netting game]
The \textit{imbalance netting game} is the cooperative game $(N,c)$ whose players are the facilitator's consumers and whose characteristic function $c(S)$ equals the imbalance cost incurred by coalition $S$'s aggregated residual under the two-price imbalance settlement. The value is not given a priori; it is computed for each coalition through the pipeline of Fig.~\ref{fig:flowchart}: a coalition quantity bid is formed (Section~\ref{sec:bidding}) and settled ex-post to obtain all coalition costs $c(S)$ (Section~\ref{sec:cost}), after which such costs are divided among consumers by an allocation mechanism (Section~\ref{sec:alloc}).
\end{definition}

The first two steps construct the game and the third solves it. We describe each in turn.

\begin{remark}\label{forc}
The facilitator generates all forecasts centrally, including renewable production and each consumer's consumption. Net-demand forecasts are then calculated as consumption minus its PPA-based share of renewable production. As a result, consumers cannot influence the game by misreporting their own forecasts.
\end{remark}

\begin{figure}[t]
    \centering
    \begin{tikzpicture}[
>=Stealth,
font=\footnotesize,
every node/.style={inner sep=1.2pt},
block/.style={
    draw,
    fill=yellow!20,
    minimum width=1.95cm,
    minimum height=0.72cm,
    line width=0.5pt,
    align=center
}
]


\node[block] (bid)   at (0,0)
{Bidding\\mechanism};

\node[block] (cost)  at (3.3,0)
{Ex-post cost\\calculation};

\node[block] (alloc) at (6.6,0)
{Allocation\\mechanism};


\path (bid.east)--(cost.west)
coordinate[midway] (m1);

\path (cost.east)--(alloc.west)
coordinate[midway] (m2);


\node[font=\scriptsize,align=center]
at ($(m1)+(0,4mm)$)
{Coalition bid\\[-0.1mm]$P^{\mathrm{DA}}_{t,S}$};

\node[font=\scriptsize,align=center]
at ($(m2)+(0,4mm)$)
{Coalition cost\\[-0.15mm]$c(S)$};


\draw[->,thick]
(bid.east)--(cost.west);

\draw[->,thick]
(cost.east)--(alloc.west);


\draw[->,thick]
($(bid.south)+(0,-0.34)$)--(bid.south);

\draw[->,thick]
($(cost.south)+(0,-0.34)$)--(cost.south);

\draw[->,thick]
(alloc.south)--($(alloc.south)+(0,-0.34)$);


\node[font=\scriptsize,align=center]
at ($(bid.south)+(0,-0.5)$)
{Forecasts};

\node[font=\scriptsize,align=center]
at ($(cost.south)+(0,-0.5)$)
{Realizations};

\node[font=\scriptsize,align=center]
at ($(alloc.south)+(0,-0.5)$)
{Individual costs $x_i$};

\end{tikzpicture}
        \vspace{-1mm}
    \caption{\small Three-step pipeline of the imbalance netting game. Individual forecasts of renewable production and consumption, made by the facilitator, feed the bidding step to produce day-ahead quantity bids $P^{\mathrm{DA}}_{t,S}$  for each coalition $S\!\subseteq\!N$. Realizations feed the ex-post cost calculation that defines the characteristic function $c(S)$, which shows the imbalance cost for the coalition $S$. The allocation mechanism divides the grand-coalition imbalance cost $c(N)$ into individual consumer costs $x_i$.}
    \label{fig:flowchart}
    \vspace{-2mm}
\end{figure}
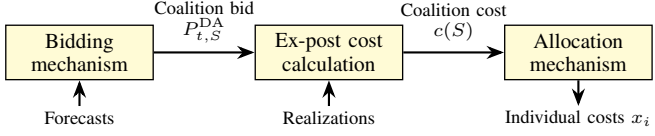

\vspace{-2mm}
\subsection{Bidding Mechanism}\label{sec:bidding}
We bid each consumer's residual individually at the newsvendor-optimal quantile that minimizes its expected imbalance cost, and the facilitator submits the aggregate as the coalition bid. We use a probabilistic bidding strategy similar to the one proposed by~\cite{pinson2007}, under which the
optimal day-ahead quantity bid (in MW) for consumer $i$ in hour $t$ is
\begin{subequations}
\begin{align}
P^{\mathrm{DA}}_{t,i} &= \big(G^{\mathrm{E}_i}_{t}\big)^{-1}(\alpha_t),
  && \forall t \in \mathcal{T},\ \forall i \in \mathcal{N}, \\
\alpha_t &= \frac{\pi^-_t}{\pi^+_t + \pi^-_t}, && \forall t \in \mathcal{T}, \label{abc}
\end{align}
\end{subequations}
where $\mathcal{T}=\{1,2,\dots,|\mathcal{T}|\}$ is the set of hourly settlement periods over the study horizon, $G^{\mathrm{E}_i}_{t}(\cdot)$ is the cumulative distribution of net-demand for consumer $i$, and $\pi^+_t$ and $\pi^-_t$ are the penalties for positive (surplus) and negative (deficit) deviation. This is a newsvendor problem: $\alpha_t$ is the critical fractile that balances the cost of over- and under-bidding, and the optimal quantity bid is the corresponding quantile of the net-demand distribution. The distribution $G^{\mathrm{E}_i}_{t}(\cdot)$ is estimated by sampling historical consumer data, and $\pi^+_t$ and $\pi^-_t$ by sampling historical price data. Because the critical fractile $\alpha_t$ depends only on prices and not on the coalition, coalition bids are formed by summing singleton bids, $P^{\mathrm{DA}}_{t,S}=\sum_{i\in S}P^{\mathrm{DA}}_{t,i}$, so that the coalition imbalance $\Delta_{t,S}$, which is mathematically defined in Section \ref{sec:cost}, is additive across consumers by construction\footnote{Bidding at coalition-level quantiles would lower $c(S)$ further, since the summed singleton bids are feasible but generally suboptimal for the coalition's newsvendor problem, but would make $\Delta_{t,i}$ coalition-dependent, undermining per-consumer decomposition. The reported netting gains therefore constitute, in expectation, a lower bound on the total value achievable through aggregation. If $\alpha=\frac{1}{2}$, i.e., the penalties for surplus and deficit are symmetric in expectation, the bid is the median, which under symmetry of distributions equals the mean and is therefore additive.}. Netting does not arise from the bid but from ex-post pricing: the coalition is settled on the sign of its aggregate imbalance $\Delta_{t,S}$, so offsetting individual deviations reduces the coalition's imbalance cost.

\subsection{Ex-post Imbalance Cost Calculation}\label{sec:cost}
Recall that the coalition  imbalance is additive, i.e.,
$\Delta_{t,S}\!=\!\sum_{i\in S}\Delta_{t,i}$. It is the coalition's day-ahead quantity
bid $P^{\mathrm{DA}}_{t,S}$ plus its contracted share $\beta_S$ of the realized
portfolio renewable production $P^{\mathrm{G}}_t$, minus its realized consumption
$P^{\mathrm{D}}_{t,S}$:
\begin{equation}\label{eq:imbalance}
    \Delta_{t,S}=P^{\mathrm{DA}}_{t,S}+\beta_S P^{\mathrm{G}}_t-P^{\mathrm{D}}_{t,S},
    \qquad \forall t\in \mathcal{T},\ \forall S\subseteq \mathcal{N}.
\end{equation}

A positive imbalance $\Delta_{t,S}\!>\!0$ means available energy exceeds consumption, a surplus, so the coalition is \textit{long}. In contrast, a negative imbalance $\Delta_{t,S}\!<\!0$ means consumption exceeds available energy, a deficit, so the coalition is \textit{short}. We write $\Delta_{t,N}$ for the \textit{grand-coalition imbalance}.

Under the two-price imbalance settlement adopted here, a coalition's imbalance is priced by whether it \textit{helps} or \textit{harms} the imbalance of the entire power system: an imbalance that opposes the system imbalance (helping) is settled at the day-ahead market price and earns nothing beyond it, while one that reinforces it (harming) is settled at a less favorable price. Specifically, if the power system is long, a power excess by a coalition is priced below the day-ahead market price, meaning the coalition has lost the opportunity to sell that volume earlier at a higher price in the day-ahead market. Conversely, if the power system is short, a power deficit by a coalition is penalized above the day-ahead market price. Under this scheme, the imbalance price can never improve on the day-ahead market price, so there is no arbitrage opportunity between day-ahead and balancing markets and hence no incentive to create imbalance intentionally.  

The bids are compared ex-post to the realized day-ahead and balancing prices, production, and demand, giving the total cost each coalition incurs in the day-ahead and balancing markets.

Under the two-price imbalance settlement, the \textit{price spread} in hour $t$ for coalition $S\subseteq N$ is
\begin{equation} \label{spread}
    \lambda_{t,S} =
    \begin{cases}
        \lambda_t^{\mathrm{DA}}-\lambda_t^+, & \text{if } \Delta_{t,S} \le 0, \\
        \lambda_t^{\mathrm{DA}}-\lambda_t^-, & \text{if } \Delta_{t,S} > 0,
    \end{cases}
    \qquad \forall t\in \mathcal{T},\ \forall S\subseteq N,
\end{equation}
where $\lambda_t^+$, $\lambda_t^-$, and $\lambda_t^{\mathrm{DA}}$ are the up-regulation, down-regulation, and day-ahead market prices, respectively. The helping direction always settles at the day-ahead market price, so $\lambda_t^+\!=\!\lambda_t^{\mathrm{DA}}$ when the power system is long and $\lambda_t^-\!=\!\lambda_t^{\mathrm{DA}}$ when it is short. The spread $\lambda_{t,S}$ is therefore zero whenever the coalition helps the power system and its product with $\Delta_{t,S}$ is non-negative, and positive only when harming. For the grand coalition, we refer to $\lambda_{t,N}$ as the \textit{grand-coalition spread}.

The penalties $\pi^+_t$ and $\pi^-_t$ in~\eqref{abc} are the expected magnitudes of the two spread branches in~\eqref{spread}, taken over the historical price sample: a positive deviation (surplus) settles at the down-regulation spread and a negative deviation (deficit) at the up-regulation spread, so $\pi^+_t=\mathbb{E}[\lambda^{\mathrm{DA}}_t-\lambda^-_t]$ and $\pi^-_t=\mathbb{E}[\lambda^+_t-\lambda^{\mathrm{DA}}_t]$, both non-negative.

\begin{remark} \label{remark1}
As the facilitator is assumed to be a price-taker, the balancing prices are fixed within each direction, so the spread a coalition faces changes only when its net position switches between helping and harming the power system, not when its imbalance volume grows.
\end{remark}

The characteristic function of coalition $S$ is the sum over the horizon $\mathcal{T}$ of its imbalance priced at the spread,
\begin{equation}\label{eq:charfun}
    c(S) = \sum_{t \in \mathcal{T}} \lambda_{t,S}\,\Delta_{t,S}, \qquad \forall S\subseteq N,
\end{equation}
which is the imbalance cost of coalition $S$. In summary, a coalition helps the power system when its imbalance opposes the system imbalance and harms it otherwise. Helping settles at the day-ahead market price, so $\lambda_{t,S} = 0$ and no imbalance cost is incurred that hour, whereas harming settles below day-ahead when the coalition is long and above it when short, so $c(S) \geq 0$ in both harming cases.\footnote{Unlike the two-price imbalance settlement, $c(S)$ can be negative under a one-price scheme, implying that imbalances can be profitable and thus creating arbitrage opportunities between the day-ahead and balancing markets.}
The game is a minimization of the imbalance cost $c(S)$. This completes the imbalance netting game $(N,c)$.

\vspace{-2mm}
\subsection{Game Properties and Definitions}\label{sec:GameProperties}
Two questions are essential: whether it is beneficial for consumers to merge into larger coalitions, and how the resulting imbalance cost should be shared among them. The remainder of this section introduces the properties that answer these two questions.
The first question is one of coalition structure: whether combining coalitions ever destroys collective value.

\begin{definition}[Sub-additivity]
A game is sub-additive if $c(S_1 \cup S_2) \leq c(S_1)+c(S_2)$ for all disjoint $S_1,S_2 \subseteq N$.
\end{definition}
Sub-additivity guarantees that a merged coalition costs no more than its parts separately, so consumers are never worse off by merging. The imbalance netting game is sub-additive because pooling lets opposing imbalances cancel and reduce the imbalance cost, while aligned imbalances sum additively, so under the price-taking assumption in Remark \ref{remark1}, the pooled imbalance cost never exceeds the sum of the consumers' standalone imbalance costs.

A stronger condition is concavity, which strengthens this from ``merging never hurts'' to ``merging into larger coalitions is increasingly attractive.'' It is defined through the marginal contribution of a consumer to the imbalance cost.

\begin{definition}[Marginal Cost Contribution]\label{MCC_def}
The marginal cost contribution of consumer $i$ to a coalition $S\!\subseteq\!N\!\setminus\!\{i\}$ is $M_i(S)\!=\!c(S \cup \{i\})\!-\!c(S)$, the change in coalition imbalance cost when $i$ joins $S$. For the grand coalition, we write $M_i\!=\!M_i(N\!\setminus\!\{i\})$.
\end{definition}

\begin{definition}[Concavity]
A game is concave if
\begin{equation}
    c(S_1 \cup S_2)+c(S_1 \cap S_2)\leq c(S_1)+c(S_2), \quad \forall S_1,S_2\subseteq N.
\end{equation}
\end{definition}
This condition has a direct reading in terms of marginal cost contributions. To obtain the equivalent marginal contribution interpretation, let $S_1\!=\!A\cup\{i\}$ and $S_2\!=\!B$, where $A\subseteq B$ and $i\notin B$. 
The definition rearranges to $M_i(B) \leq M_i(A)$. Concavity therefore requires that a consumer's marginal contribution to imbalance cost does not increase as the coalition it joins grows, so that adding a consumer to a larger coalition saves at least as much as adding it to a smaller one. It is a desirable property because it guarantees that a \textit{stable} allocation exists, a notion made precise later.

The imbalance netting game in our setting, though sub-additive, is \textit{not} concave. Let us provide a counterexample: Suppose the power system is short, and consider coalitions $A \subseteq B$ with $\Delta_A\!=\!-2$~MW (short, harming) and $\Delta_B\!=\!+1$~MW (long, helping). Adding a long consumer $i$ with $\Delta_i\!=\!+1$~MW offsets part of $A$'s shortfall, lowering $A$'s imbalance cost and giving a strictly negative $M_i(A)$. However, $i$ only deepens $B$'s already-helping long position, leaving $B$'s imbalance cost at zero and giving $M_i(B)\!=\!0$. Hence, $M_i(B)\!>\!M_i(A)$ despite $A\!\subseteq\!B$, so the game is not concave. 
This distinction is central to the paper: stability does not follow from concavity and must instead be checked for each allocation mechanism separately, which is the subject of Table~\ref{tab:mechanisms} in Section~\ref{sec:alloc}.

Recall the second question on how to allocate the coalition cost among the consumers. An allocation vector $\mathbf{x}^S\!=\!\{x_i^S\!\mid\!i\in\!S\}$ assigns a cost share $x_i^S$ to each consumer $i$ in coalition $S$. In particular, the \textit{grand-coalition allocation} is denoted by $\mathbf{x}^N\!=\!\{x_i^N\!\mid\!i\in\!N\}$. For notational simplicity, we omit the superscript $N$ when referring to the grand-coalition allocation and write $\mathbf{x}\!=\!\{x_i\!\mid\!i\in\!N\}$.
For an allocation to hold a coalition together, two minimal conditions must hold: no consumer should do worse than on its own, and the full coalition cost should be distributed.

\begin{definition}[Individual Rationality]
An allocation vector $x^S$ is individually rational if $x_i^S\leq c(\{i\})$ for all $i\in S$.
\end{definition}

\begin{definition}[Collective Rationality]
An allocation vector $x^S$ is collectively rational if $\sum_{i \in S} x_i^S\!=\!c(S)$.
\end{definition}

\begin{definition}[Imputation]
An allocation vector is an imputation if it is both individually and collectively rational.
\end{definition}

Individual and collective rationality are necessary but not sufficient: an imputation can still leave a subgroup of consumers better off on their own, in which case the grand coalition is \textit{unstable}. To detect this, we measure how much a coalition could gain by breaking away.

\begin{definition}[Excess]
The excess of a coalition $S\subseteq N$ under the grand-coalition allocation $\mathbf{x}$ is
\begin{equation}\label{excess1}
    \epsilon(S,\mathbf{x})=c(S)-\sum_{i \in S}x_{i}.
\end{equation}
\end{definition}

A negative excess means that a coalition can reduce its total charge relative to what it is charged under the grand-coalition allocation $\mathbf{x}$, and therefore has an incentive to deviate. Requiring the excess of every coalition to be non-negative gives rise to the central notion of \textit{stability} used throughout this paper.

\begin{definition}[Core]
The core is the set of imputations with $\epsilon(S,\mathbf{x})\geq 0$ for all $S \subseteq N$:
\begin{equation}
C = \left\{ \mathbf{x} \,\middle|\, \sum_{i \in N} x_i = c(N), \; \epsilon(S,\mathbf{x}) \geq 0, \; \forall S \subseteq N \right\}.
\end{equation}
\end{definition}
Intuitively, the core is the feasible region of allocations that distribute the entire grand-coalition value and under which no coalition, from a single consumer to any subgroup, can do better by breaking away. An allocation in the core is therefore stable, since no group has an incentive to leave; we adopt such stability as our notion of a well-behaved allocation. The core can be empty, and even when it is non-empty, computing an allocation inside it is non-trivial. This is the task of an \textit{allocation mechanism}, which maps the coalitions' characteristic functions to a single allocation. 

\vspace{-2mm}
\section{Imbalance cost allocation mechanisms}\label{sec:alloc}
An allocation mechanism is the third step of the pipeline in Fig.~\ref{fig:flowchart}: it maps the characteristic function to a single allocation $\mathbf{x}$ that divides the grand-coalition cost among consumers. We study six allocation mechanisms: the Shapley value mechanism, the marginal cost contribution mechanism, the \ac{vcg} mechanism, the nucleolus mechanism, the marginal price allocation mechanism, and finally the Gately point mechanism. Hereafter, allocations produced by these mechanisms are denoted using superscripts. For example, the Shapley value allocation is denoted by $\mathbf{x}^{\rm{SV}}\!=\!\{x^{\rm{SV}}_i\!\mid\!i\in\!N\}$.

\vspace{-3mm}
\subsection{Shapley value} 
The Shapley value \cite{Shapley1953} assigns each consumer $i$ a weighted average of its marginal cost contribution $M_i(S)\!=\!c(S\cup\{i\})\!-\!c(S)$ over all coalitions $S$ not containing $i$:
\begin{equation}
    x^{\rm{SV}}_i=\sum_{S\subseteq N\setminus\{i\}}\frac{|S|!\,(|N|-|S|-1)!}{|N|!}\,M_i(S), \ \ \forall i\in N,
\end{equation}
where $|S|$ and $|N|$ denote the number of consumers in coalition $S$ and the grand coalition $N$, respectively.
The Shapley value is the unique allocation satisfying the standard fairness axioms of cooperative game theory \cite{Shapley1953}. For concave games, the Shapley value belongs to the core \cite{Shapley1971}; however, recall from Section \ref{sec:GameProperties} that the imbalance netting game is non-concave. Therefore, this guarantee does not apply in the present setting. Moreover, computing the Shapley value requires evaluating the imbalance cost for every possible coalition, resulting in $2^{|N|}$ coalition cost evaluations.

\vspace{-2mm}
\subsection{Marginal Cost Contribution (MCC)}\label{sec:mcc} This mechanism charges each consumer $i$ its marginal cost contribution to the grand coalition, as defined in Definition \ref{MCC_def}:
\begin{equation}
    x^{\rm{MCC}}_i = M_i = c(N) - c(N\setminus\{i\}), \qquad \forall i\in N,
    \label{eq:MCC}
\end{equation}
that is, the additional imbalance cost incurred by the grand coalition due to the inclusion of consumer $i$. Computing the MCC allocation requires evaluating the imbalance costs of only the $|N|$ leave-one-out coalitions and the grand coalition, for a total of $|N|+1$ coalition cost evaluations. This is computationally much more efficient than the Shapley value.

\vspace{-2mm}
\subsection{VCG}\label{sec:VCG}
The \ac{vcg} mechanism \cite{Vickrey1961,Clarke1971,Groves1973} assigns to each consumer $i$ the externality that its participation imposes on the \emph{remaining} consumers in the grand coalition:
\begin{subequations}
\begin{align}
    x^{\rm{VCG}}_i &= \sum_{t \in \mathcal{T}} \lambda_{t,N}\,\Delta_{t,N\setminus\{i\}} - c(N\setminus\{i\}), \ \ \forall i\in N,
    \label{eq:VCG_others_cost}
\end{align}
where $c(N\setminus\{i\}) = \sum_{t \in \mathcal{T}} \lambda_{t,N\setminus\{i\}}\,\Delta_{t,N\setminus\{i\}}$ 
is the imbalance cost incurred by the remaining consumers when consumer $i$ is absent. The first term represents the imbalance cost of the remaining consumers when consumer i is present, obtained by settling their imbalance $\Delta_{t,N\setminus\{i\}}$ at the grand-coalition spread $\lambda_{t,N}$\footnote{Standard VCG requires the cost borne by the remaining consumers when consumer i participates. We price the remaining consumers' imbalance $\Delta_{t,N\setminus i}$ at the grand coalition spread $\lambda_{t,N}$ as a modeling choice. This equals the marginal price allocation \ref{eq:mp}, which yields the identity \ref{eq:vcg-mcc-relations}.}. Equivalently,
\begin{align}
    x^{\rm{VCG}}_i &= \sum_{t \in \mathcal{T}} \bigl(\lambda_{t,N}-\lambda_{t,N\setminus\{i\}}\bigr)\,\Delta_{t,N\setminus\{i\}}, \ \ \forall i\in N,
    \label{eq:VCG_others_cost1}
\end{align}
\end{subequations}
showing that consumer $i$ receives a non-zero allocation only in hours where its participation changes the imbalance price spread faced by the remaining consumers, that is, where $\lambda_{t,N}\neq\lambda_{t,N\setminus\{i\}}$. Under the price-taking assumption of the facilitator (Remark~\ref{remark1}), this occurs only when the participation of consumer $i$ changes the sign of the portfolio's  imbalance. We refer to such a consumer as a \textit{pivotal consumer}; otherwise, consumer $i$'s allocation is zero. As shown in Proposition \ref{prop:mcc_vcg_group_rational}, VCG is never positive for this game. In practice, this means that most consumers receive a zero allocation while those that receive a non-zero allocation are paid, leaving the facilitator with no income and a budget deficit, which is a well-known drawback of the VCG mechanism \cite{VCG13}. As with the MCC mechanism, the VCG mechanism requires evaluating only the grand coalition and the $|N|$ leave-one-out coalitions, for a total of $|N|+1$ coalition cost evaluations.


\vspace{-2mm}
\subsection{Nucleolus} The nucleolus \cite{Scmeidler1969} selects the imputation that lexicographically maximizes the vector of coalition excesses $\epsilon(S,\mathbf{x})$ sorted in non-decreasing order. With the excess definition in \eqref{excess1}, the smallest excess corresponds to the coalition with the strongest incentive to deviate. The nucleolus therefore first makes this smallest excess as large as possible, then the second smallest, and so on. It has no closed form and is computed through a sequence of linear programs involving constraints for all coalitions. Since all $2^{|N|}$ coalition values are required as constraints, the method becomes computationally challenging for large portfolios. Whenever the core is non-empty, the nucleolus lies in it, so we use it as a stability benchmark.

\vspace{-2mm}
\subsection{Marginal Price}\label{sec:mp} The marginal price mechanism \cite{Guo2021}, a uniform-price scheme related to the shadow price imputation of \cite{Zuluaga2025}, charges each consumer its own imbalance at the grand-coalition spread each hour:
\begin{equation}\label{eq:mp}
    x^{\rm{MP}}_i=\sum_{t\in \mathcal{T}}\lambda_{t,N}\,\Delta_{t,i}, \qquad \forall i\in N.
\end{equation}

Unlike the two-price imbalance settlement, where each consumer's imbalance is settled according to whether it helps or harms the power system, this mechanism settles every consumer's imbalance at the single grand-coalition imbalance price spread, $\lambda_{t,N}$, regardless of the consumer's own imbalance direction. It therefore constitutes an internal one-price settlement scheme. In time periods where the facilitator's portfolio helps the power system, $\lambda_{t,N}\!=\!0$, so no consumer is charged for that period. Since the imbalance cost is computed only for the grand coalition, this mechanism is computationally the most efficient among those considered.

\vspace{-3mm}
\subsection{Gately point}

The Gately point\cite{Gately1974} allocates according to each consumer's \emph{propensity to disrupt}, 
seeking to balance departure incentives across consumers. In a cost-minimization setting, 
the propensity to disrupt of consumer $i$ under an imputation $x$ is defined as:

\begin{subequations}\label{eq:gately}
\begin{equation}\label{eq:gately-d}
d_i(\mathbf{x})=\frac{\sum_{j\neq i}x_j-c(N\setminus\{i\})}{x_i-c(\{i\})},
\end{equation}
where the numerator is the negative of the cost savings others gain from $i$'s  participation, and the denominator is the negative of the cost savings $i$ gains  from being in the coalition. Equivalently, $d_i$ is the ratio of cost savings for  the others to cost savings for $i$ itself. When both benefit from aggregation, this  ratio is positive ($d_i > 0$), indicating that $i$ is a valuable contributor, i.e., the  others benefit from including $i$ relative to $i$'s own benefit.

The Gately point is the imputation that minimizes the largest propensity to disrupt across consumers,
\begin{equation}\label{eq:gately-dstar}
    \mathbf{x}^{\rm{GP}} = \arg\min_{\mathbf{x}^{\rm{GP}} \in \mathcal{X}} \, \max_{i\in N} \, d_i(\mathbf{x}),
\end{equation}
where $\mathcal{X}$ is the set of imputations of $N$. This minimization balances the incentives to depart. It admits a closed form \cite{Littlechild1976} at the allocation where all propensities to disrupt are equal, with common value
\begin{equation}\label{equal}
    d^*=\frac{\sum_{j\in N}M_j-c(N)}{c(N)-\sum_{j\in N}c(\{j\})},
\end{equation}
which, substituted into \eqref{eq:gately-d}, gives each consumer's allocation
\begin{equation}\label{eq:gat_alloc}
    x^{\rm{GP}}_i=\frac{d^*\,c(\{i\})+M_i}{d^*+1}, \qquad \forall i\in N.
\end{equation}
\end{subequations}

The Gately point requires $2|N|+1$ coalition costs, namely $c(N)$, the $|N|$ leave-one-out 
costs $c(N\!\setminus\!\{i\})$, and the $|N|$ singletons $c(\{i\})$, so it scales linearly 
with portfolio size.

\vspace{-2mm}
\subsection{Properties and Comparison}\label{sec:properties_and_comparison}
We compare mechanisms on four properties \cite{MITRIDATI2021102177,ACourseInCooperativeGameTheory}. \textit{Computational complexity} is the number of coalition costs $c(S)$ that must be evaluated, together with the effort to form the allocation from them, as a function of portfolio size. \textit{Budget balance} requires the consumer charges to sum exactly to $c(N)$, leaving the facilitator with neither surplus nor deficit. \textit{Group rationality} requires every coalition to be at least as well off in the grand coalition as by splitting off, so that no subgroup has an incentive to leave. \textit{Additivity} requires the allocation of the whole game to equal the sum of the allocations of its per-period, e.g., hourly, subgames. \textit{Incentive compatibility} is not compared: by Remark~\ref{forc} all forecasts are made centrally by the facilitator, so consumers cannot misreport. Together with the assumption of inelastic demand, centralized forecasting removes strategic reporting or dispatch decisions by consumers; incentive compatibility is therefore not an active consideration in the present model.

Table~\ref{tab:mechanisms} summarizes the six mechanisms with respect to these properties. The Shapley value and the nucleolus both require the costs of all $2^{|N|}$ coalitions. In addition, computing the nucleolus may require solving up to $2^{|N|}-1$ linear programs, rendering both mechanisms computationally intractable as the portfolio grows. Moreover, because the game is non-concave, the Shapley value is not guaranteed to lie in the core.

VCG is computationally inexpensive but, except in the degenerate case $c(N)=0$, is not budget-balanced in this game and yields a budget deficit in favor of the consumers. This deficit follows directly from additive bidding. Under the per-consumer newsvendor bidding described in Section~\ref{sec:bidding}, imbalances are additive across consumers, so $\Delta_{t,N\setminus\{i\}}\!=\!\Delta_{t,N}\!-\!\Delta_{t,i}$ for every $t\!\in\!\mathcal{T}$. Substituting this identity into~\eqref{eq:VCG_others_cost} and using $c(N)=\sum_t\lambda_{t,N}\Delta_{t,N}$ from~\eqref{eq:charfun} yields
\begin{subequations}\label{eq:vcg-mcc-relations}
\begin{equation}
x_i^{\rm{VCG}} = x_i^{\rm{MCC}} - x_i^{\rm{MP}}, \qquad \forall i \in N,
\label{eq:vcg-mcc-mp}
\end{equation}
where $x_i^{\rm{MCC}}$ is the marginal cost contribution defined in~\eqref{eq:MCC}, and $x_i^{\rm{MP}}$ is the marginal price allocation defined in~\eqref{eq:mp}. Because the marginal price allocation is budget-balanced, summing~\eqref{eq:vcg-mcc-mp} over all consumers gives
\begin{equation}
\sum_{i \in N} x_i^{\rm{VCG}} = \sum_{i \in N} x_i^{\rm{MCC}} - c(N).
\label{eq:vcg-mcc-sum}
\end{equation}
\end{subequations}

Consequently, the total amount collected under VCG is lower than that collected under MCC by exactly $c(N)$. Equivalently, the VCG budget deficit exceeds the MCC budget deficit by exactly $c(N)$. This result is consistent with VCG assigning non-zero transfers only to pivotal consumers in this setting. The identity follows from the linear, hourly price-spread structure of the imbalance netting game in~\eqref{eq:charfun} and does not hold for cooperative games in general.

\begin{table}[t!]
\centering
\caption{\small Studied mechanisms and their properties. Results specific to the imbalance netting game are cited where applicable.}
\label{tab:mechanisms}
\footnotesize
\resizebox{\columnwidth}{!}{
    \begin{threeparttable}
    \setlength{\tabcolsep}{4pt}
    \renewcommand{\arraystretch}{1.25}
    \begin{tabular}{@{}lcccc@{}}
    \toprule
    Mechanism & Complexity & \begin{tabular}[b]{@{}c@{}}Budget\\ balanced\end{tabular} & \begin{tabular}[b]{@{}c@{}}Group\\ rational\end{tabular} & Additive \\
    \midrule
    Shapley value~\cite{Shapley1953}  & $2^{|N|}$   & \checkmark          & --\,\tnote{a}          & \checkmark \\
    MCC\tnote{e}   & $|N|{+}1$   & --\,\tnote{b}       & \checkmark~(Prop..~\ref{prop:mcc_vcg_group_rational})   & \checkmark~(Lem.~\ref{lem:MCCAdd}) \\
    VCG~\cite{Vickrey1961,Clarke1971,Groves1973} & $|N|{+}1$   & --\,\tnote{c}       & \checkmark~(Prop.~\ref{prop:mcc_vcg_group_rational})   & \checkmark~(Cor. D.3) \\
    Nucleolus~\cite{Scmeidler1969}      & $2^{|N|}$   & \checkmark          & \checkmark             & --  \\
    Marginal price~\cite{Guo2021,Zuluaga2025} & $1$         & \checkmark          & \checkmark\,\tnote{d}  & \checkmark \\
    Gately point~\cite{Gately1974,Littlechild1976}   & $2|N|{+}1$  & \checkmark          & --            & --  \\
    \bottomrule
    \end{tabular}
    \begin{tablenotes}[flushleft]
    \footnotesize
    \item[a] Only guaranteed to be in the core for concave games \cite{Shapley1971}.
    \item[b] Budget balanced if and only if \eqref{eq:BudCon1} holds (Theorem~\ref{the:MCCFinal}); otherwise the facilitator is left in deficit.
    \item[c] Never budget balanced for this game: the deficit always favors consumers. See Sections \ref{sec:VCG} and \ref{sec:properties_and_comparison}.
    \item[d] Established for this class of games in \cite{Guo2021}.
    \item[e] MCC is a marginal contribution rule; its complexity and properties are derived directly rather than cited to prior work.
    \end{tablenotes}
    \end{threeparttable}
}
\vspace{-2mm}
\end{table}

\section{Analytical Results}\label{sec:analytical}
This section presents three analytical results for the imbalance netting game. First, we derive a necessary and sufficient condition for MCC to be budget-balanced, supporting the corresponding entry in Table~\ref{tab:mechanisms}. Second, we establish group rationality for MCC and VCG. Third, we characterize the existence of the Gately point and express it as a combination of MCC and singleton costs. These results constitute the analytical contributions of this paper.

\subsection{Budget Balance of MCC}\label{sec:VCGBB}
MCC is generally not budget-balanced, and whether the budget imbalance favors the participants or the mechanism operator depends on the setting. We show that for the imbalance netting game it always favors the consumers: the charges collected fall short of the imbalance cost, leaving the facilitator in deficit.

The MCC allocation charges each consumer its marginal contribution to the imbalance cost. Budget balance therefore holds exactly when the consumers' marginal cost contributions sum to the grand-coalition cost $c(N)$, a property we call \textit{additivity of marginal contributions at the top level }(AMC).
\begin{definition}[AMC]
    A game exhibits AMC if the grand-coalition cost equals the sum of the consumers' marginal cost contributions:
\begin{equation}\label{AMC}
    c(N)=\sum_{i\in N}M_i.
\end{equation}
\end{definition}

We first show that MCC is additive over subgames. Let $\mathrm{MCC}(c)\!=\!\bigl(x^{\rm{MCC}}_1,\dots,x^{\rm{MCC}}_{|N|}\bigr)$ denote the allocation vector that MCC assigns in a game $c$, with $x^{\rm{MCC}}_i\!=\!c(N)\!-\!c(N\!\setminus\!\{i\})$.

\begin{lemma}\label{lem:MCCAdd}
Write the game as a sum of hourly subgames, $c = \sum_{t\in\mathcal{T}} c_t$, where $c_t$ is the game for imbalance realization $t$.\footnote{This decomposition relies on the game being temporally separable: from \eqref{eq:charfun}, $c(S)=\sum_{t\in\mathcal{T}}\lambda_{t,S}\Delta_{t,S}$ is a sum of independent per-hour terms, so each hour forms its own subgame. Intertemporal coupling, such as storage carrying state of charge across hours, would prevent this decomposition and the additivity in question would no longer apply.} MCC is additive over this decomposition:
\begin{equation}
    \mathrm{MCC}(c) = \sum_{t \in \mathcal{T}} \mathrm{MCC}(c_t),
\end{equation}
where the sum is taken componentwise over consumers.
\end{lemma}

\begin{proof} See Appendix~\ref{AppA}. \end{proof}

The additive decomposition lets budget balance be checked one hour at a time and then summed over the horizon. We first give a sufficient condition.

\begin{lemma}\label{lem:BBholds}
AMC, and hence budget balance of MCC, holds if
\begin{equation}\label{eq:BudCon1}
    \lambda_{t,N} = \lambda_{t,N\setminus\{i\}}, \qquad \forall i\in N,\ \forall t\in \mathcal{T},
\end{equation}
that is, if removing any single consumer $i$ leaves the grand-coalition price spread $\lambda_{t,N}$ of \eqref{spread} unchanged in every hour.
\end{lemma}

\begin{proof}
    See Appendix~\ref{AppB}.
\end{proof}

The condition can fail, and when it does the failure is one-directional.

\begin{lemma}\label{lem:BBbreaks}
    If \eqref{eq:BudCon1} does not hold, MCC returns a non-budget-balanced allocation that always leaves the facilitator in budget deficit.
\end{lemma}

\begin{proof}
    See Appendix~\ref{AppC}.
\end{proof}

Together, Lemmas~\ref{lem:BBholds} and \ref{lem:BBbreaks} characterize budget balance for MCC.

\begin{theorem}\label{the:MCCFinal}
The imbalance netting game under the MCC mechanism is budget-balanced if and only if \eqref{eq:BudCon1} holds. Otherwise, the consumer charges do not cover the imbalance cost the facilitator bears, leaving it in deficit:
\begin{equation} \label{xxc}
    \sum_{i \in N} x^{\rm{MCC}}_i < c(N).
\end{equation}
\end{theorem}
\begin{proof}
    Lemma~\ref{lem:BBholds} gives budget balance under \eqref{eq:BudCon1}, and Lemma~\ref{lem:BBbreaks} gives the deficit when it fails.
\end{proof}
\subsection{Group Rationality of MCC and VCG}
\begin{proposition}\label{prop:mcc_vcg_group_rational}
    For every consumer $i\in N$, $x_i^{\rm{MCC}}\leq x_i^{\rm{MP}}$ and $x_i^{\rm{VCG}}\leq 0$. Hence both MCC and VCG are group rational as MCC always allocates a lower cost than a group rational mechanism and under VCG no consumer ever makes a positive net payment.
\end{proposition}
\begin{proof}
    See Appendix \ref{app:gr_mcc_vcg}.
\end{proof}

\subsection{Existence and Uniqueness of the Gately point}\label{sec:GatelyExistence}
The Gately point is not defined for every cooperative game, as it requires the game to be \textit{essential}~\cite{ACourseInCooperativeGameTheory}, meaning that cooperation yields a cost saving relative to consumers acting individually. This section characterizes when the Gately point exists for the imbalance netting game and proposes an alternative allocation for cases in which it does not.

\begin{definition}[Well-defined Gately point]
The Gately point is said to be \textit{well-defined} for a cooperative game if the imputation solving \eqref{eq:gately-dstar} exists and is unique, or equivalently that \eqref{equal} is well-defined. 
\end{definition}
For the Gately point itself to be well-defined, not every individual propensity to disrupt in \eqref{eq:gately-d} needs to be well-defined. A consumer for whom the individual propensity is undefined is instead charged its standalone cost \cite{Staudacher2019}.
\begin{lemma}
    For every $i \in N$, $c(\{i\}) \ge M_i$. Moreover, if the game is essential, i.e.,
    \[
        c(N) < \sum_{j\in N} c(\{j\}),
    \]
    then $c(\{i\}) > M_i$ for at least one $i \in N$.
\end{lemma}
\begin{proof}
    See Appendix~\ref{app:gat}.
\end{proof}

\begin{theorem}\label{thm:GatelyWellDefined}
The Gately point is well-defined for the imbalance netting game if and only if:
\begin{subequations}
\begin{equation}\label{eq:gat_a}
        c(N) < \sum_{i\in N}c(\{i\}),
\end{equation}
which holds if and only if there exists at least one hour $t\in \mathcal{T}$ with $\lambda_{t,N}\neq 0$ such that
\begin{equation}\label{eq:gat_b}
        \left|\sum_{i\in N}\Delta_{t,i}\right| < \sum_{i\in N}|\Delta_{t,i}|.
\end{equation}
\end{subequations}

Moreover, whenever \eqref{eq:gat_a}--\eqref{eq:gat_b} hold, the Gately point is the unique imputation given by \eqref{equal}--\eqref{eq:gat_alloc}.
\end{theorem}
\begin{proof}
    See Appendix~\ref{app:gat}.
\end{proof}

\begin{remark}
The numerator of \eqref{equal} is $\sum_j M_j - c(N)$, i.e., equal to the MCC deficit of Theorem~\ref{the:MCCFinal}, and its denominator is the negative of the aggregation saving $\sum_j c(\{j\})-c(N)$ of Theorem~\ref{thm:GatelyWellDefined}. Hence, $d^{*}\ge 0$ equals the share of the aggregation gain that is not recovered under MCC, and \eqref{eq:gat_alloc} reads $x_i^{\mathrm{GP}} = \theta c(\{i\}) + (1-\theta)x_i^{\mathrm{MCC}}$ with $\theta = d^{*}/(1+d^{*})$. The Gately point blends the MCC charges with the standalone costs in the proportion that restores budget balance, and
coincides with MCC exactly when \eqref{eq:BudCon1} holds.
\end{remark}

\begin{corollary}
If instead consumers' imbalances are perfectly aligned in every hour $t \in \mathcal{T}$, i.e.,
\[
    \left|\sum_{i\in N} \Delta_{t,i}\right|
    = \sum_{i\in N} |\Delta_{t,i}|, \qquad \forall t \in \mathcal{T},
\]
the game becomes inessential, and hence the unique core allocation is
\begin{equation}
    x_i = c(\{i\}), \qquad \forall i \in N,
\end{equation}
which charges each consumer its standalone imbalance cost.
\end{corollary}
\begin{proof}
    See Appendix~\ref{app:gat}.
\end{proof}

\section{Numerical Results} \label{sec:results}
This section applies the imbalance netting game to 19 Danish consumers sharing a photovoltaic PPA, quantifies aggregation gains, assesses allocation stability, and empirically verifies the analytical results of Section~\ref{sec:analytical}.

\begin{figure}[b]
    \centering
    \input{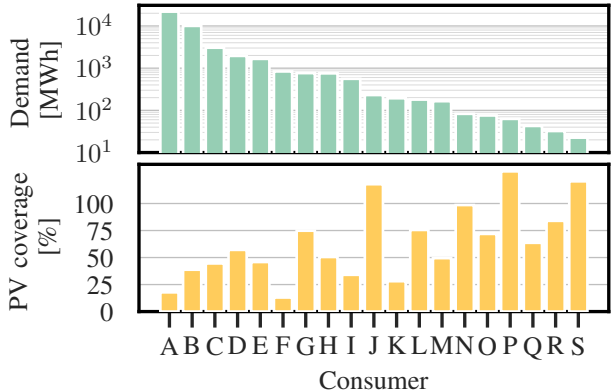}
    \vspace{-8mm}
    \captionsetup{skip=4pt}
    \caption{\small{Bar chart showing consumer demand and PV coverage during the case study period. Demand is plotted on a logarithmic scale.}}
    \label{fig:demand}
\end{figure}

\vspace{-1mm}
\subsection{Input Data}
The theoretical findings are applied to a portfolio of 19 consumers, labelled $\rm{A}$ to $\rm{S}$, each holding a \ac{ppa} with a photovoltaic (PV) plant through the facilitator. The case study uses real-world hourly consumption measurements, electricity prices, PV production, and PV production forecasts. The latter three datasets are obtained from \textit{Energinet}'s Energi Data Service~\cite{EnergiDataService} while consumption data is from Reel. All data is from the DK1 price zone. Source code is available in our public repository~\cite{eriksen2026code}.

Each consumer has a PPA covering an individually negotiated share of the plant's PV production. Because all consumers are supplied by the same plant, production uncertainty is common to the portfolio and cannot be reduced through netting. We study 2024, comprising $|\mathcal{T}|=8784$ hourly settlement periods.

The portfolio is highly heterogeneous, with distinct demand profiles and substantial differences in total demand and PV coverage, as shown in Fig.~\ref{fig:demand}. Consumer $\rm{A}$ has a higher total demand than all other consumers combined, while consumer $\rm{B}$ exceeds consumers $\rm{C}$--$\rm{S}$ combined. Their imbalances are therefore more likely to determine the aggregate portfolio imbalance, making them potential pivotal consumers. PV coverage is defined as total PPA energy divided by total consumption and thus does not capture their temporal alignment. This real-world heterogeneity provides a suitable test case for revealing the benefits and limitations of the six cooperative cost-allocation mechanisms.

\vspace{-1mm}
\subsection{Results: Cost Savings}\label{sec:results_savings}
To motivate the application of cooperative game theory to imbalance netting within portfolios, the gains from creating coalitions are calculated in Fig. \ref{fig:Gains}. We define the \textit{coalition allocation ratio} as $\frac{c(S)}{\sum_{i\in S}c(\{i\})}$, that is, the coalition's cost divided by the sum of the standalone costs of its members. Fig. \ref{fig:Gains} shows that the average coalition allocation ratio decreases with coalition size reaching 10\% for the grand coalition, indicating substantial gains from aggregation. However, the variation across coalitions of equal size shows that coalition composition also matters. Some coalitions achieve cost savings exceeding 20\%, substantially above the average saving for coalitions of the same size. The challenge is then how to divide this coalitional cost saving among consumers in a stable manner.

\begin{figure}[!t]
    \centering
    \input{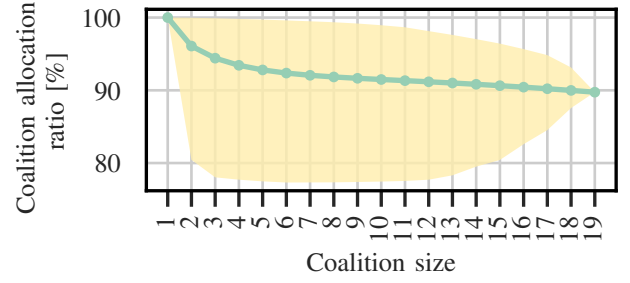}
    \vspace{-8mm}
    \captionsetup{skip=4pt}
    \caption{\small{The gain from aggregation in terms of cost savings over the year, captured by the coalition allocation ratio, as a function of coalition size, ranging from 1 (no coalition) to 19 (grand coalition). For each coalition size, there may be multiple possible coalitions. The green line represents the average allocation ratio across all possible coalitions of a given size, while the yellow shaded area indicates the variation in the ratio among these coalitions.}}
    \label{fig:Gains}
\end{figure} 

\begin{figure*}[t]
    \centering
    \input{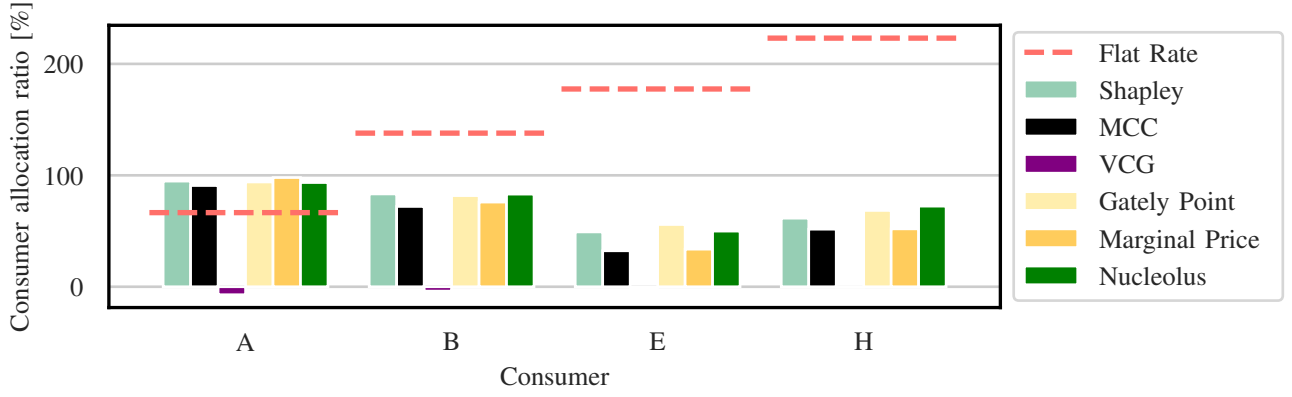}
    \vspace{-8mm}
    \captionsetup{skip=3.8pt}
    \caption{\small{Consumer allocation ratio for some selected consumers.}}
    \label{fig:costRatio}
    \vspace{-1.5mm}
\end{figure*}

To take an individual perspective, the \emph{consumer allocation ratio} for every consumer \textit{i} is defined as $\frac{x^{\rm{M}}_i}{c(\{i\})}$, where $\rm{M}$ is the underlying allocation mechanism.
This metric shows the relative cost of consumer $i$ as part of the grand coalition. The results for some selected consumers are shown in Fig. \ref{fig:costRatio}, while the full results are available in Appendix \ref{app:full_ratio}. Any consumer with a consumer allocation ratio at most 100\% will prefer to be in the grand coalition, compared to operating alone. For comparison as a naive approach, the \emph{Flat Rate Allocation} is defined as dividing grand coalition costs proportional to consumption: $x^{\rm{FR}}_i=c(N) \sum_{t\in\mathcal{T}}\frac{P_{t,i}^{\rm{D}}}{\sum_{i\in N}P_{t,i}^{\rm{D}}}$.

Fig.~\ref{fig:costRatio} shows that the flat-rate allocation violates individual rationality, charging some consumers more than 200\% of their standalone imbalance cost and thereby giving them an incentive to leave the portfolio. The largest consumer, $\rm{A}$, has a lower allocation ratio under the flat rate than under the stable mechanisms, whereas most other consumers have a higher ratio. The flat rate therefore causes smaller consumers to subsidize consumer $\rm{A}$, as it disregards both $\rm{A}$'s large imbalance volume and its frequent role in determining the portfolio imbalance direction. By contrast, the budget-balanced mechanisms yield similar allocation ratios, all markedly different from those under the flat rate. Consumer costs thus depend far more on replacing the flat rate than on the choice among the budget-balanced mechanisms. This finding motivates a practical hybrid billing approach that combines partial socialization through a flat rate with partial individualization through one of the six cost-allocation mechanisms, as discussed in Appendix~\ref{BusinessModel}.

\begin{table}[b]
\centering
\caption{\small{Minimum coalition excess under each allocation mechanism (\EUR{}).}}
\resizebox{\columnwidth}{!}{
    \begin{tabular}{|c|c|c|c|c|c|c|}
    \hline
    Shapley & MCC & VCG & Nucleolus & \makecell{Marginal\\price} & \makecell{Gately\\point} & \makecell{Flat\\rate} \\
    \hline
    3.2 & 29.2 & 84.6 & 14.6 & 0 & 10.2 & -32,486 \\
    \hline
    \end{tabular}
}
\vspace{-3mm}
\label{tab:ExcessCost}
\end{table}

\vspace{-3mm}
\subsection{Results: Stability of Mechanisms}\label{sec:results_stability}
Table~\ref{tab:ExcessCost} reports the minimum coalition excess under each mechanism, corresponding to the coalition with the strongest incentive to deviate. A non-negative minimum excess implies that all coalition rationality inequalities are satisfied. For a budget-balanced allocation, it additionally implies core membership and hence stability. All budget-balanced game-theoretic mechanisms are stable for this dataset, whereas the flat-rate allocation is not. MCC and VCG also satisfy all coalition rationality inequalities, but are not imputations because they do not recover the grand-coalition cost. Among the stable, budget-balanced mechanisms, the marginal price allocation has the smallest stability margin. The nucleolus provides the lexicographically most stable benchmark, while the Gately point has the largest minimum excess among the computationally tractable, budget-balanced mechanisms.
\subsection{Empirical Verification of Analytical Results}
We numerically verify Theorems~\ref{the:MCCFinal} and~\ref{thm:GatelyWellDefined} using the case study data. The budget balance condition~\eqref{eq:BudCon1} fails in any hour in which removing at least one consumer flips the sign of the portfolio imbalance. Table~\ref{tab:signflip} reports the corresponding split across all $|\mathcal{T}|$=8{,}784 hours of 2024, together with the resulting MCC and VCG deficits.

\begin{table}[b] 
\centering 
\caption{\small{Sign-flip hours and the resulting MCC and VCG budget deficits. Deficits are computed as $\sum_{i\in N}x_i-c(N)$ over each hour subset; negative values indicate under-recovery by the facilitator. $c(N)$ = \EUR{179{,}981}.}} \label{tab:signflip} 
\begin{tabular}{lccc} 
\toprule & \makecell{Hours\\(share)} & \makecell{MCC deficit\\ (\EUR{})} & \makecell{VCG deficit\\ (\EUR{})} \\ 
\midrule \makecell[l]{No pivotal consumer} & 6{,}733 (76.65\%) & 0.00 & $-159{,}304$ \\ 
\makecell[l]{At least one pivotal\\consumer} & 2{,}051 (23.35\%) & $-10{,}995$ & $-31{,}672$ \\ 
\midrule \textbf{Total} & \textbf{8{,}784 (100\%)} & $\mathbf{-10{,}995}$ & $\mathbf{-190{,}976}$ \\ \bottomrule 
\end{tabular} 
\vspace{-2mm}
\end{table} 

Consistent with Lemma~\ref{lem:BBholds}, MCC is exactly budget-balanced in the 76.65\% of hours satisfying~\eqref{eq:BudCon1}. Its entire shortfall of \EUR{10,995} arises in the remaining 23.35\% of hours, in which removing at least one consumer flips the portfolio imbalance sign. This confirms Lemma~\ref{lem:BBbreaks} and Theorem~\ref{the:MCCFinal}. As noted in Section~\ref{sec:properties_and_comparison}, VCG incurs a deficit even when MCC is budget-balanced.

By Theorem~\ref{thm:GatelyWellDefined}, the Gately point is well-defined and unique if the dataset contains at least one hour with mixed consumer imbalance signs. Such hours account for 86.77\% of the dataset. The Gately point is therefore well-defined and unique for 2024 and for any subset containing at least one of these 7{,}622 hours.

We also verify Proposition \ref{prop:mcc_vcg_group_rational}. The inequality $x_i^{\rm{MCC}}\leq x_i^{\rm{MP}}$ in Lemma~\ref{lem:mcc-group-rational} holds for all 19 consumers, with an average gap of \EUR{578.68}. Because the marginal price allocation is group rational, this confirms the group rationality of MCC. Likewise, Corollary~\ref{cor:vcg-group-rational} is verified: all VCG allocations range from $-$\EUR{9,548} to \EUR{0}, so no consumer makes a positive net payment under VCG. Consumer $\rm{A}$ receives a VCG transfer of \EUR{9,548} and is pivotal in 1{,}372 of the 8{,}784 hours (15.6\%), whereas consumer $\rm{S}$ receives a negligible transfer and is pivotal in only two hours (0.02\%). This confirms the observation in Section~\ref{sec:VCG} that VCG assigns non-zero transfers only to consumers whose participation changes the portfolio imbalance direction.

Finally, we verify~\eqref{eq:vcg-mcc-mp} numerically: the residual $x_i^{\rm{VCG}}-(x_i^{\rm{MCC}}-x_i^{\rm{MP}})$ is zero for every consumer, validating both the identity and its implementation.

\section{Conclusion}\label{sec:conclusion}
We formulate imbalance cost allocation within a facilitator's portfolio as a cooperative \textit{imbalance netting game}, in which coalition costs arise from an asymmetric two-price settlement under inelastic demand and additive bidding. In a case study of 19 consumers using 2024 Danish data, aggregation reduces the grand-coalition imbalance cost by 10\% relative to the sum of the consumers' standalone costs. The distribution of these savings, however, depends strongly on the allocation mechanism. A flat-rate allocation proportional to consumption violates individual rationality, charging some consumers more than twice their standalone imbalance cost. By contrast, all budget-balanced game-theoretic mechanisms examined belong to the core for this dataset. Among them, the Gately point and marginal price allocations combine stability and budget balance with computational requirements that scale linearly with portfolio size. The MCC and VCG mechanisms satisfy all group rationality inequalities but under-recover the grand-coalition cost and therefore do not constitute imputations. 

We establish the game-specific properties underlying these findings, including a necessary and sufficient condition for budget balance of the MCC mechanism, group rationality of the MCC and VCG mechanisms, and a necessary and sufficient condition for the Gately point to be well-defined and unique. Future work should examine symmetric one-price settlement, under which pooling primarily reduces cost volatility rather than expected cost. Further extensions include decentralized forecasting and flexible demand, which introduce incentive-compatibility considerations, as well as coalition-level bidding and price-making portfolios, in which portfolio imbalances can influence balancing price spreads.

\section*{Acknowledgement}
The authors would like to thank Reel and, in particular, David Ribberholt Ipsen and Edoardo Simioni, for their assistance in understanding the business model of a facilitator and for supplying data. We would also like to thank Farzaneh Pourahmadi and Valdemar Søgaard for providing feedback on early versions of the manuscript. This manuscript includes grammatical corrections assisted by Claude AI and Mistral AI and code has been partially generated through the same models. All intellectual content, interpretation, and conclusions are the authors' own.

\vspace{-1mm}
\bibliographystyle{IEEEtran}
\bibliography{tex/bibliography/Bibliography}

\begin{thebibliography}{10}
\providecommand{\url}[1]{#1}
\csname url@samestyle\endcsname
\providecommand{\newblock}{\relax}
\providecommand{\bibinfo}[2]{#2}
\providecommand{\BIBentrySTDinterwordspacing}{\spaceskip=0pt\relax}
\providecommand{\BIBentryALTinterwordstretchfactor}{4}
\providecommand{\BIBentryALTinterwordspacing}{\spaceskip=\fontdimen2\font plus
\BIBentryALTinterwordstretchfactor\fontdimen3\font minus \fontdimen4\font\relax}
\providecommand{\BIBforeignlanguage}[2]{{%
\expandafter\ifx\csname l@#1\endcsname\relax
\typeout{** WARNING: IEEEtran.bst: No hyphenation pattern has been}%
\typeout{** loaded for the language `#1'. Using the pattern for}%
\typeout{** the default language instead.}%
\else
\language=\csname l@#1\endcsname
\fi
#2}}
\providecommand{\BIBdecl}{\relax}
\BIBdecl

\bibitem{ImbalanceDirective}
\BIBentryALTinterwordspacing
{European Union}, ``{Directive-2019/944},'' 2019. [Online]. Available: \url{https://eur-lex.europa.eu/eli/dir/2019/944/oj/eng}
\BIBentrySTDinterwordspacing

\bibitem{EnerginetFullCost}
\BIBentryALTinterwordspacing
``Energinet position on full cost balancing,'' Energinet, Tech. Rep., 2024. [Online]. Available: \url{https://energinet.dk/media/xbpjtgde/energinet-position-on-full-cost-balancing.pdf}
\BIBentrySTDinterwordspacing

\bibitem{Chu2016}
S.~Chu, Y.~Cui, and N.~Liu, ``The path towards sustainable energy,'' \emph{Nature Materials}, vol.~16, no.~1, pp. 16--22, 2016.

\bibitem{Fleten2018}
S.~E. Fleten, K.~T. Midthun, T.~Bj{\o}rkvoll, A.~Werner, and M.~Fodstad, ``The portfolio perspective in electricity generation and market operations,'' in \emph{Proceedings of the 15th International Conference on the European Energy Market (EEM)}, Lodz, Poland, 2018.

\bibitem{Tekani2025}
D.~V. Tekani, J.~Shi, and H.~Grebel, ``Managing renewable energy resources using equity-market risk tools -- the efficient frontiers,'' \emph{Energy Efficiency}, vol.~18, no.~6, pp. 1--10, 2025.

\bibitem{Oren2025}
S.~S. Oren and A.~Papalexopoulos, ``{VPP} integration and market participation: {B}eyond optimization,'' \emph{IEEE Power and Energy Magazine}, vol.~23, no.~6, pp. 166--172, 2025.

\bibitem{MITRIDATI2021102177}
L.~Mitridati, J.~Kazempour, and P.~Pinson, ``Design and game-theoretic analysis of community-based market mechanisms in heat and electricity systems,'' \emph{Omega}, vol.~99, p. 102177, 2021.

\bibitem{vespermann2021}
N.~Vespermann, T.~Hamacher, and J.~Kazempour, ``Access economy for storage in energy communities,'' \emph{IEEE Transactions on Power Systems}, vol.~36, no.~3, pp. 2234--2250, 2021.

\bibitem{Exizidis2019}
L.~Exizidis, J.~Kazempour, A.~Papakonstantinou, P.~Pinson, Z.~D. Gr\`eve, and F.~Vall\'ee, ``Incentive-compatibility in a two-stage stochastic electricity market with high wind power penetration,'' \emph{IEEE Transactions on Power Systems}, vol.~34, no.~4, pp. 2846--2858, 2019.

\bibitem{Baeyens2013}
E.~Baeyens, E.~Y. Bitar, P.~P. Khargonekar, and K.~Poolla, ``Coalitional aggregation of wind power,'' \emph{IEEE Transactions on Power Systems}, vol.~28, no.~4, pp. 3774--3784, 2013.

\bibitem{Russo2020}
M.~Russo and V.~Bertsch, ``A looming revolution: {I}mplications of self-generation for the risk exposure of retailers,'' \emph{Energy Economics}, vol.~92, p. 104970, 2020.

\bibitem{Guo2021}
Y.~Guo, M.~Pan, and Y.~Gong, ``Aggregation-based colocation datacenter energy management in wholesale markets,'' \emph{IEEE Transactions on Cloud Computing}, vol.~9, no.~1, pp. 66--78, 2021.

\bibitem{Zuluaga2025}
T.~{Valencia Zuluaga} and S.~S. Oren, ``Stable and fair uniform price allocations of community choice aggregation gains in retail electricity markets,'' \emph{Energy Systems}, pp. 1--34, 2025.

\bibitem{ACourseInCooperativeGameTheory}
S.~R. Chakravarty, M.~Mitra, and P.~Sarkar, \emph{A Course on Cooperative Game Theory}.\hskip 1em plus 0.5em minus 0.4em\relax Cambridge University Press, 2015.

\bibitem{pinson2007}
P.~Pinson, C.~Chevallier, and G.~N. Kariniotakis, ``Trading wind generation from short-term probabilistic forecasts of wind power,'' \emph{IEEE Transactions on Power Systems}, vol.~22, no.~3, pp. 1148--1156, 2007.

\bibitem{Shapley1953}
L.~S. Shapley, ``A value for n-person games,'' in \emph{Contributions to the Theory of Games II}.\hskip 1em plus 0.5em minus 0.4em\relax Princeton, NJ: Princeton University Press, 1953, pp. 307--317.

\bibitem{Shapley1971}
------, ``Cores of convex games,'' \emph{International Journal of Game Theory}, vol.~1, no.~1, pp. 11--26, 1971.

\bibitem{Vickrey1961}
W.~Vickrey, ``Counterspeculation, auctions, and competitive sealed tenders,'' \emph{The Journal of Finance}, vol.~16, no.~1, pp. 8--37, 1961.

\bibitem{Clarke1971}
E.~H. Clarke, ``Multipart pricing of public goods,'' \emph{Public Choice}, vol.~11, no.~1, pp. 17--33, 1971.

\bibitem{Groves1973}
T.~Groves, ``Incentives in teams,'' \emph{Econometrica}, vol.~41, no.~4, pp. 617--631, 1973.

\bibitem{VCG13}
M.~H. Rothkopf, ``Thirteen reasons why the {V}ickrey-{C}larke-{G}roves process is not practical,'' \emph{Oper. Res.}, vol.~55, no.~2, pp. 191--197, 2007.

\bibitem{Scmeidler1969}
D.~Schmeidler, ``The nucleolus of a characteristic function game,'' \emph{SIAM Journal on Applied Mathematics}, vol.~17, no.~6, pp. 1163--1170, 1969.

\bibitem{Gately1974}
D.~Gately, ``Sharing the gains from regional cooperation: A game theoretic application to planning investment in electric power,'' \emph{International Economic Review}, vol.~15, no.~1, pp. 195--208, 1974.

\bibitem{Littlechild1976}
S.~C. Littlechild and K.~G. Vaidya, ``The propensity to disrupt and the disruption nucleolus of a characteristic function game,'' \emph{International Journal of Game Theory}, vol.~5, no. 2-3, pp. 151--161, 1976.

\bibitem{Staudacher2019}
\BIBentryALTinterwordspacing
J.~Staudacher and J.~Anwander, ``Conditions for the uniqueness of the gately point for cooperative games,'' 2019. [Online]. Available: \url{https://arxiv.org/abs/1901.01485}
\BIBentrySTDinterwordspacing

\bibitem{EnergiDataService}
\BIBentryALTinterwordspacing
Energinet, ``Energi data service,'' licensed under Creative Commons CC-BY 4.0. [Online]. Available: \url{https://energidataservice.dk/}
\BIBentrySTDinterwordspacing

\bibitem{eriksen2026code}
A.~W. Eriksen, ``Code for ``sharing the gains of aggregation'','' \url{https://github.com/AsmusWE/Aggregation-imbalance-sharing}, 2026, gitHub repository; code in release v1.0.

\end{thebibliography}

\clearpage

\appendix

\subsection{Additivity of MCC}\label{AppA}
\begin{proof}
The imbalance netting game decomposes over the horizon: defining the hourly subgame $c_t(S) = \lambda_{t,S}\Delta_{t,S}$, we have $c(S) = \sum_{t\in\mathcal{T}} c_t(S)$ for every $S\subseteq N$ by \eqref{eq:charfun}. The payment to consumer $i$ is its marginal cost contribution. Substituting the decomposition,
\begin{align}
    x^{\rm{MCC}}_i 
    &= c(N) - c(N \setminus \{i\}) \nonumber\\
    &= \sum_{t\in\mathcal{T}} c_t(N) - \sum_{t\in\mathcal{T}} c_t(N \setminus \{i\}) \nonumber\\
    &= \sum_{t\in\mathcal{T}} \bigl[ c_t(N) - c_t(N \setminus \{i\}) \bigr] = \sum_{t\in\mathcal{T}} M_{t,i},
    \label{eq:MCCadditive}
\end{align}
where $M_{t,i} = c_t(N)\!-\!c_t(N\!\setminus\!\{i\})$ is consumer $i$'s marginal cost contribution in hour $t$. Thus, each consumer's MCC charge is the sum of its hourly marginal cost contributions, i.e., $\mathrm{MCC}(c)\!=\!\sum_{t\in\mathcal{T}} \mathrm{MCC}(c_t)$ componentwise, proving Lemma~\ref{lem:MCCAdd}.
\end{proof}

\subsection{Conditions for Budget Balance of MCC}\label{AppB}
\begin{proof}
Consider a single-hour subgame $c_t$ of $c$. By \eqref{eq:charfun}, its characteristic function is
\begin{subequations}\label{eq:vcgcon}
\begin{equation}\label{eq:vcgcon-uN}
    c_t(N) = \lambda_{t,N}\,\Delta_{t,N} = \lambda_{t,N}\sum_{i\in N} \Delta_{t,i}.
\end{equation}

Assume that in this hour, the price spread is unchanged when any single consumer leaves the coalition,
\begin{equation}\label{eq:vcgcon-cond}
    \lambda_{t,N} = \lambda_{t,N\setminus\{i\}}, \qquad \forall i\in N.
\end{equation}

Under this condition, $c_t(N\!\setminus\!\{i\})\!=\!\lambda_{t,N\!\setminus\!\{i\}}\Delta_{t,N\setminus\{i\}}\!=\!\lambda_{t,N}\Delta_{t,N\!\setminus\!\{i\}}$, so the marginal contribution of consumer $i$ is
\begin{align}\label{eq:vcgcon-Mi}
    M_{t,i}
    &= c_t(N) - c_t(N\setminus\{i\}) \nonumber\\
    &= \lambda_{t,N}\Delta_{t,N} - \lambda_{t,N}\Delta_{t,N\setminus\{i\}}
     = \lambda_{t,N}\Delta_{t,i},
\end{align}
using $\Delta_{t,N}\!-\!\Delta_{t,N\!\setminus\!\{i\}}\!=\!\Delta_{t,i}$. Summing over consumers,
\begin{equation}\label{eq:vcgcon-sum}
    \sum_{i\in N} M_{t,i}
    = \lambda_{t,N}\sum_{i\in N}\Delta_{t,i}
    = c_t(N).
\end{equation}
\end{subequations}

Hence, AMC as defined in \eqref{AMC} holds for hour $t$. Since MCC allocates marginal contributions, AMC implies that the consumer charges sum to $c_t(N)$, so MCC is budget-balanced in that hour. By the additivity of Lemma~\ref{lem:MCCAdd}, budget balance in every hour gives budget balance over the whole horizon, proving Lemma~\ref{lem:BBholds}.
\end{proof}

\subsection{Budget Deficit of MCC}\label{AppC}
\begin{proof}
Consider an hour $t$ with
\begin{equation}\label{eq:deltaset}
    \delta_t = \{\, i \in N \mid \lambda_{t,N} \ne \lambda_{t,N\setminus\{i\}} \,\}\neq\emptyset .
\end{equation}

The facilitator and every consumer within the portfolio are price-takers, so removing a consumer changes the spread $\lambda_{t,N}$ only by flipping the sign of the coalition imbalance (see Remark \ref{remark1}). Hence, each $i\!\in\!\delta_t$ has $\Delta_{t,i}$ sharing the sign of $\Delta_{t,N}$ with $|\Delta_{t,i}|\!>\!|\Delta_{t,N}|$, so its removal flips the coalition between helping and harming the power system. A consumer $i\!\notin\!\delta_t$ leaves the spread unchanged and satisfies $M_{t,i}\!=\!\lambda_{t,N}\Delta_{t,i}$, as in Appendix~\ref{AppB}. Two cases arise, according to the direction of the grand coalition in hour $t$.

\emph{Case 1: the grand coalition harms}, so $c_t(N)\!=\!\lambda_{t,N}\Delta_{t,N}\!>\!0$. Removing any $i\!\in\!\delta_t$ flips the remainder to helping, which settles at the day-ahead price, so $c_t(N\!\setminus\!\{i\})\!=\!0$ and thus $M_{t,i}\!=\!c_t(N)$. Summing over all consumers and using the identity $\lambda_{t,N}\sum_{i\in N}\Delta_{t,i}\!=\!c_t(N)$ to replace the non-flipping terms,
\begin{subequations}
\begin{align}\label{eq:mcc-case1}
    \sum_{i\in N} M_{t,i}
    &= |\delta_t|\,c_t(N) + \lambda_{t,N}\!\!\sum_{i\notin\delta_t}\!\!\Delta_{t,i} \nonumber\\
    &= (|\delta_t|+1)\,c_t(N) - \lambda_{t,N}\!\!\sum_{i\in\delta_t}\!\!\Delta_{t,i},
\end{align}
where $|\delta_t|$ is the number of consumers in $\delta_t$. For each $i\in\delta_t$, $\Delta_{t,i}$ shares the sign of $\Delta_{t,N}$ and exceeds it in magnitude, so $\lambda_{t,N} \Delta_{t,i} > \lambda_{t,N}\Delta_{t,N} = c_t(N)$; hence $\lambda_{t,N}\sum_{i\in\delta_t}\Delta_{t,i} > |\delta_t|\, c_t(N)$. \emph{This inequality $M_{t,i} < \lambda_{t,N}\Delta_{t,i}$ holds termwise for each individual $i \in \delta_t$, not merely in aggregate, a fact used directly in the proof of Lemma~\ref{lem:mcc-group-rational}}. Substituting into \eqref{eq:mcc-case1},
\begin{equation}\label{eq:mcc-case1-deficit}
    \sum_{i\in N} M_{t,i} < (|\delta_t|+1)\,c_t(N) - |\delta_t|\,c_t(N) = c_t(N).
\end{equation}

\emph{Case 2: the grand coalition helps}, so $c_t(N)\!=\!0$ and $\lambda_{t,N}\!=\!0$. Every non-flipping consumer then has $M_{t,i}=\lambda_{t,N}\Delta_{t,i}\!=\!0$, while removing any $i\!\in\!\delta_t$ flips the remainder to harming, giving $c_t(N\setminus\{i\})\!>\!0$ and $M_{t,i}\!=\!c_t(N)-c_t(N\setminus\{i\})\!=\!-c_t(N\setminus\{i\})$. Therefore
\begin{equation}\label{eq:mcc-case2-deficit}
    \sum_{i\in N} M_{t,i}
    = -\!\!\sum_{i\in\delta_t}\!\! c_t(N\setminus\{i\})
    < 0 = c_t(N).
\end{equation}
\end{subequations}

In both cases, $\sum_{i\in N} M_{t,i} < c_t(N)$ whenever $\delta_t\ne\emptyset$. Since MCC allocates $M_{t,i}$ to each consumer, the charges collected in such an hour sum to strictly less than the hourly imbalance cost, leaving the facilitator in budget deficit. The result extends to any game containing such an hour, so MCC is not budget-balanced whenever \eqref{eq:BudCon1} fails, and the shortfall always favors the consumers.
\end{proof}

\subsection{Group Rationality of MCC and VCG}\label{app:gr_mcc_vcg}
\setcounter{theorem}{0} 
\renewcommand{\thetheorem}{D.\arabic{theorem}}
\begin{lemma}[Group rationality of MCC] \label{lem:mcc-group-rational}
For every consumer $i \in N$,
\begin{subequations}
\begin{equation}
x_i^{\mathrm{MCC}} \leq x_i^{\mathrm{MP}}, \label{eq:mcc-leq-mp}
\end{equation}
where $x_i^{\mathrm{MP}}$ is the marginal price allocation of \eqref{eq:mp}. Consequently, MCC is group rational: since marginal price is group rational for this class of games~\cite{Guo2021}, for every coalition $S \subseteq N$,
\begin{equation}
\sum_{i \in S} x_i^{\mathrm{MCC}} \leq \sum_{i \in S} x_i^{\mathrm{MP}} \leq c(S), \label{eq:mcc-group-rational}
\end{equation}
\end{subequations}
i.e., $\epsilon(S, x^{\mathrm{MCC}}) \geq 0$ for all $S \subseteq N$, so no coalition can lower its total charge by leaving the grand coalition under MCC.
\end{lemma}

\begin{proof}
By construction, $x_i^{\mathrm{MP}} = \sum_{t\in\mathcal{T}} \lambda_{t,N}\Delta_{t,i}$. Appendix~C (proof of Lemma~IV.3) shows termwise that $M_{t,i} = \lambda_{t,N}\Delta_{t,i}$ for $i \notin \delta_t$, and $M_{t,i} < \lambda_{t,N}\Delta_{t,i}$ for $i \in \delta_t$. Summing over $t \in \mathcal{T}$ gives $x_i^{\mathrm{MCC}} \leq x_i^{\mathrm{MP}}$ for every $i \in N$, proving \eqref{eq:mcc-leq-mp}. Inequality~\eqref{eq:mcc-group-rational} then follows by summing over $i \in S$ and invoking the group rationality of marginal price~\cite{Guo2021}.
\end{proof}

\begin{corollary}[Group rationality of VCG]
\label{cor:vcg-group-rational}
The VCG allocation is group rational.
\end{corollary}

\begin{proof}
By \eqref{eq:vcg-mcc-mp}, $x_i^{\mathrm{VCG}} = x_i^{\mathrm{MCC}} - x_i^{\mathrm{MP}}$ for all $i \in N$. Lemma~\ref{lem:mcc-group-rational} gives $x_i^{\mathrm{MCC}} \leq x_i^{\mathrm{MP}}$, so $x_i^{\mathrm{VCG}} \leq 0$ for every consumer --- under VCG, no consumer is ever charged a positive net amount in this game. Since $c(S) \geq 0$ for every coalition $S$ (Section~II-B), it follows immediately that for any $S \subseteq N$,
\begin{equation}
\sum_{i \in S} x_i^{\mathrm{VCG}} \leq 0 \leq c(S),
\end{equation}
so $\epsilon(S, x^{\mathrm{VCG}}) \geq 0$ for all $S \subseteq N$, establishing group rationality.
\end{proof}

\begin{corollary}[Additivity of VCG]
\label{cor:vcg-additive}
The VCG mechanism is additive.
\end{corollary}

\begin{proof}
By \eqref{eq:vcg-mcc-mp}, $x_i^{\mathrm{VCG}} = x_i^{\mathrm{MCC}} -
x_i^{\mathrm{MP}}$ for all $i \in N$. MCC is additive
(Lemma~\ref{lem:MCCAdd}) and marginal price is additive by
construction, so their difference is additive over the same
per-period decomposition, proving the claim.
\end{proof}

\subsection{Existence, Well-Definedness, and Uniqueness of the Gately point}\label{app:gat}
\begin{proof}
By \eqref{equal}, the Gately point exists precisely when
$c(N) - \sum_{j\in N} c(\{j\}) \neq 0$. By sub-additivity, as defined in Definition 2, we have $c(N) \le \sum_{j\in N} c(\{j\})$; the denominator of \eqref{equal} is therefore never positive, and the Gately point exists if and only if
\begin{equation}\label{eq:gat_negative}
    c(N) < \sum_{j\in N} c(\{j\}).
\end{equation}

Consider an hour $t \in T$. By the two-price settlement (Section \ref{sec:cost}) and price-taking assumptions, the hourly cost $c_t(S)$ of any coalition $S$ depends on whether $S$ is long or short that hour, and is proportional to $(\Delta_{t,S})^+$ or $(\Delta_{t,S})^-$ accordingly, with the same proportionality constant for every $S$ (Remark \ref{remark1}). Either way, $\sum_i c_t(\{i\}) \ge c_t(N)$, with equality unless the consumers' imbalances mix signs that hour, i.e.,\ unless $\bigl|\sum_i \Delta_{t,i}\bigr|<\sum_i |\Delta_{t,i}|$; when they do, any consumer $i$ whose imbalance opposes the portfolio's satisfies
$c_t(\{i\}) + c_t(N\setminus\{i\}) > c_t(N)$, and since every other hour contributes a non-negative term to the same sum, $c(\{i\}) + c(N\setminus\{i\}) > c(N)$ overall, i.e., $c(\{i\}) > M_i$ strictly.

This, together with $c(\{i\}) \ge M_i$ for every $i$ from sub-additivity, are exactly the conditions of Staudacher and Anwander \cite{Staudacher2019} guaranteeing that whenever \eqref{eq:gat_negative} holds, the equal-propensity solution \eqref{equal} is the unique imputation.

Conversely, if $\bigl|\sum_i \Delta_{t,i}\bigr| = \sum_i |\Delta_{t,i}|$ for every $t \in T$, every hourly gap is zero, so $c(S) = \sum_{i\in S} c(\{i\})$ for all $S \subseteq N$: the game is inessential, and the unique core allocation is $x_i^{\mathrm{GP}} = c(\{i\})$ for all $i \in N$.
\end{proof}

\subsection{Full Allocation Results}\label{app:full_ratio}
Fig. \ref{fig:app_full_ratio} reports the allocation ratio for all consumers. The results show a general tendency for the flat-rate allocation to benefit larger consumers while disadvantaging smaller ones.
\begin{figure}[]
    \centering
    \input{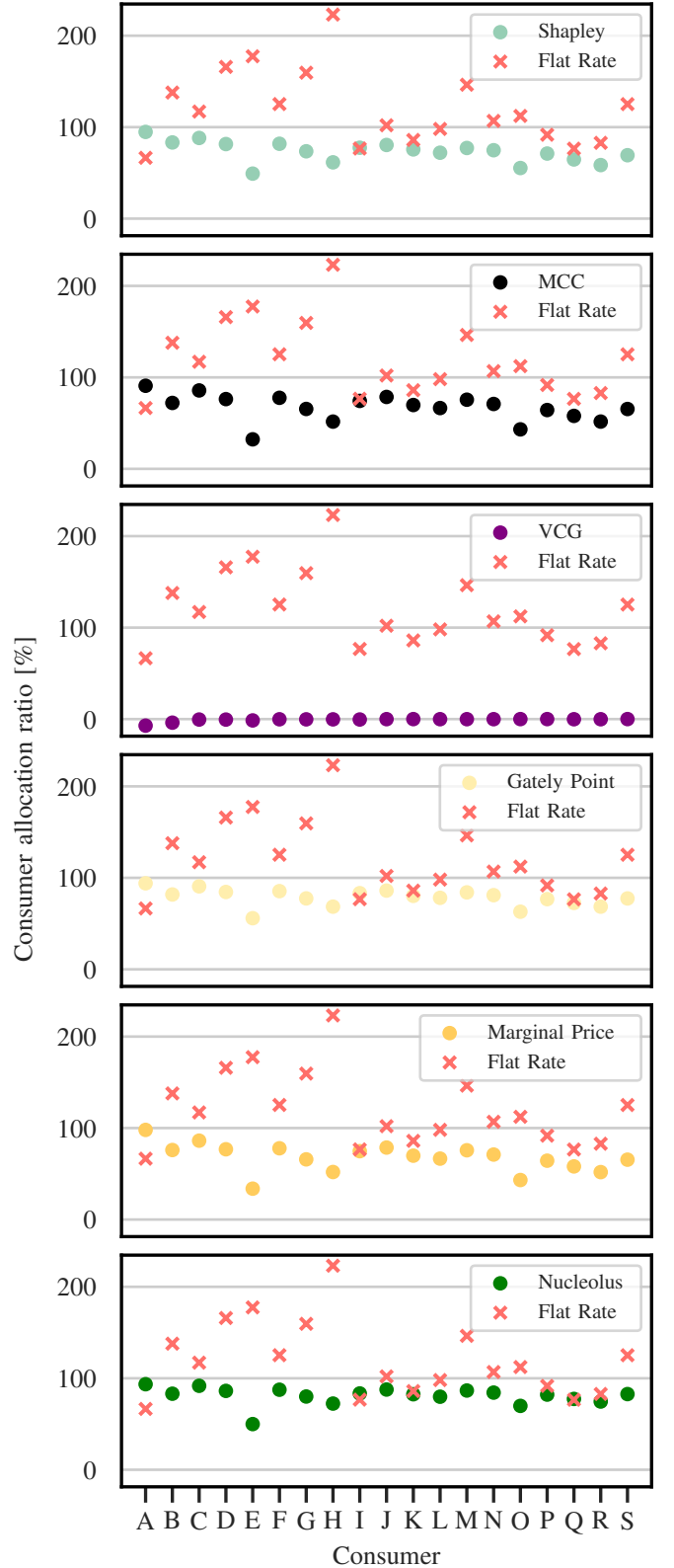}
    \vspace{-8mm}
    \captionsetup{skip=4pt}
    \caption{\small{Plot of ratio between allocated cost and singleton cost for all consumers.}}
    \label{fig:app_full_ratio}
\end{figure}

\subsection{Practical Billing and Cost Individualization}\label{BusinessModel}
We introduce the \textit{individualization grade} $\alpha^{\rm ind}\in[0,1]$ as a contractual parameter determining how the difference between an expected socialized imbalance charge and a realized individualized allocation is shared between consumers and the facilitator. Under full individualization $(\alpha^{\rm ind}=1)$, each consumer ultimately pays its realized grand-coalition allocation $x^N_i$, so the balancing risk, reduced through aggregation, is passed through to consumers individually. Under full socialization $(\alpha^{\rm ind}=0)$, consumers pay an expected socialized imbalance charge, while the facilitator bears the difference between the collected socialized charges and the realized portfolio imbalance cost. Full socialization, however, does not differentiate consumers according to the heterogeneity in netting contributions established in Sections~\ref{sec:results_savings} and~\ref{sec:results_stability}. We formalize this choice through a monthly billing rule. The cost of consumer $i$ in month $m$ is 
\begin{subequations} 
\begin{equation} 
\begin{split} C_{i,m} =\;
& \sum_{t\in \mathcal{T}_m} k_t P^{\rm D}_{t,i} + \left(1-\alpha^{\rm ind}\right) \hat{\lambda}^{\rm soc}_{m} \sum_{t\in \mathcal{T}_m} P^{\rm D}_{t,i} \\ 
&+ \alpha^{\rm ind} \left( \hat{\lambda}^{\rm ind}_{i,m} \sum_{t\in \mathcal{T}_m}P^{\rm D}_{t,i} + Q_{i,m} \right), 
\end{split} 
\label{eq:consumer_cost} 
\end{equation} 
where $k_t$ is the agreed per-unit rate covering all cost components other than imbalance, such as the day-ahead market price and overhead. Furthermore, $\hat{\lambda}^{\rm soc}_{m}$ and $\hat{\lambda}^{\rm ind}_{i,m}$ are the expected socialized and individualized imbalance costs per MWh, respectively, for month $m$.\footnote{These quantities are not imbalance price spreads. They are estimates of expected imbalance costs and therefore do not depend on the realized deviation.} Recall that $P^{\rm D}_{t,i}$ denotes consumer $i$'s consumption in period $t$. In addition, $\mathcal{T}_m$ is the set of settlement periods in month $m$. The reconciliation term is 
\begin{equation} 
Q_{i,m} = x^N_{i,m} - \hat{\lambda}^{\rm ind}_{i,m} \sum_{t\in \mathcal{T}_m}P^{\rm D}_{t,i}, \label{eq:reconciliation} 
\end{equation} 
where $x^N_{i,m}$ is the realized grand-coalition allocation obtained by applying the mechanism of Section~\ref{sec:alloc} to the data for month $m$.\footnote{For simplicity, we assume that $x^N_{i,m}$ is available when the bill for month $m$ is issued. In practice, reconciliation may occur in a subsequent billing cycle because imbalance prices may not be known when the initial bill is issued.} Substituting \eqref{eq:reconciliation} into \eqref{eq:consumer_cost} shows that the imbalance-related payment of consumer $i$ is 
\begin{equation} 
\left(1-\alpha^{\rm ind}\right) \hat{\lambda}^{\rm soc}_{m} \sum_{t\in \mathcal{T}_m}P^{\rm D}_{t,i} + \alpha^{\rm ind}x^N_{i,m}. \label{eq:hybrid_imbalance_charge} 
\end{equation} 
\end{subequations}

Thus, the billing rule interpolates between the expected socialized charge and the realized individualized allocation. At $\alpha^{\rm ind}=1$, the reconciliation term closes the full gap between the expected individualized charge and $x^N_{i,m}$. At $\alpha^{\rm ind}=0$, no reconciliation occurs, and the facilitator bears the difference between the collected socialized charges and the realized portfolio imbalance cost. For intermediate values of $\alpha^{\rm ind}$, this difference is shared between the facilitator and consumers. Individualization therefore does not eliminate balancing risk but reallocates it between the facilitator and consumers. Whether the resulting allocation is stable depends on the individualization grade. Full individualization removes the facilitator's direct financial incentive to minimize portfolio imbalance because the realized imbalance cost is passed through to consumers regardless of the facilitator's bidding performance. Any remaining incentive is indirect and arises through market competition. Partial individualization preserves some direct exposure for the facilitator while differentiating consumer charges according to their netting contributions. In an extended setting with flexible demand or consumer-provided operational information, a higher $\alpha^{\rm ind}$ could also strengthen incentives to report outages promptly or adjust flexible consumption. Such behavioral responses are outside the inelastic-demand setting studied in this paper.

Fig.~\ref{fig:costInd} illustrates this trade-off for the Gately point allocation, which provides the largest minimum coalition excess among the computationally tractable, budget-balanced mechanisms in the case study. The resulting hybrid allocation is group rational only for individualization grades above 0.70. Below this threshold, at least one coalition could reduce its total charge by leaving the grand coalition and adopting a fully individualized allocation. The individualization grade is therefore not merely a risk-sharing parameter; it also affects coalition stability. In this case study, the facilitator cannot socialize more than 30\% of the imbalance charge without violating at least one group rationality inequality.

\begin{figure}[t]
    \centering
    \input{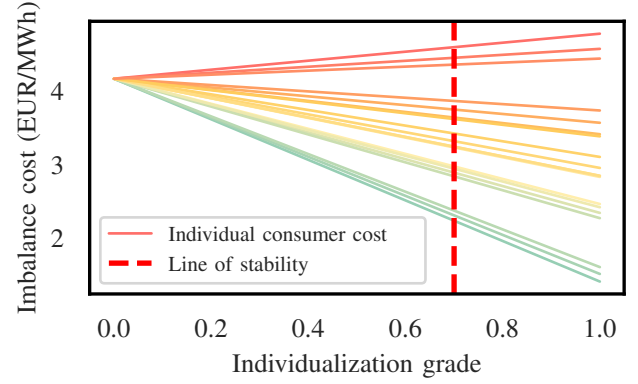}
    \caption{\small{Consumer-level imbalance charges under the Gately point allocation as a function of the individualization grade. Charges are expressed per MWh of consumption. The dashed line marks the minimum individualization grade at which the resulting allocation is group rational.}}
    \label{fig:costInd}
\end{figure}

\end{document}